\documentclass[lettersize,journal]{IEEEtran}
\usepackage{amsmath,amsfonts}
\usepackage{array}
\usepackage{textcomp}
\usepackage{stfloats}
\usepackage{url}
\usepackage{verbatim}
\usepackage{graphicx}
\usepackage{cite}
\usepackage{amsthm}
\usepackage{subcaption}
\newtheorem{definition}{Definition}
\newtheorem{remark}{Remark}
\newtheorem{lemma}{Lemma}
\newtheorem{proposition}{Proposition}
\newtheorem{corollary}{Corollary}
\newtheorem{assumption}{Assumption}
\DeclareMathOperator*{\argmin}{arg\,min}
\usepackage{soul}
\usepackage{xcolor}
\usepackage[linesnumbered,ruled,vlined]{algorithm2e}

\usepackage[switch]{lineno}
\usepackage{multirow}
\usepackage{balance}

\begin{document}




\title{FeMLoc: Federated Meta-learning for Adaptive Wireless Indoor Localization Tasks in IoT Networks}

\author{Yaya Etiabi,~\IEEEmembership{Student Member,~IEEE},
Wafa Njima,~\IEEEmembership{Member,~IEEE},
El Mehdi Amhoud,~\IEEEmembership{Member,~IEEE},

\thanks{ Yaya Etiabi is with the College of Computing, Mohammed VI Polytechnic University, Benguerir, Maroc (e-mail: yaya.etiabi@um6p.ma).}
\thanks{ Wafa Njima is with the LISITE Laboratory, Institut Supérieur d'Électronique de Paris, Paris, France (e-mail: wafa.njima@isep.fr).}
\thanks{ El Mehdi Amhoud is with the College of Computing, Mohammed VI Polytechnic University, Benguerir, Maroc (e-mail: elmehdi.amhoud@um6p.ma).}

}



\maketitle

\begin{abstract}
The rapid growth of the Internet of Things fosters collaboration among connected devices for tasks like indoor localization.
However, existing indoor localization solutions struggle with dynamic and harsh conditions, requiring extensive data collection and environment-specific calibration. These factors impede cooperation, scalability, and the utilization of prior research efforts.
To address these challenges, we propose FeMLoc, a federated meta-learning framework for localization.
FeMLoc operates in two stages: (i) collaborative meta-training where a global meta-model is created by training on diverse localization datasets from edge devices. (ii) Rapid adaptation for new environments, where the pre-trained global meta-model initializes the localization model, requiring only minimal fine-tuning with a small amount of new data.
In this paper, we provide a detailed technical overview of FeMLoc, highlighting its unique approach to privacy-preserving meta-learning in the context of indoor localization.
Our performance evaluation demonstrates the superiority of FeMLoc over state-of-the-art methods, enabling swift adaptation to new indoor environments with reduced calibration effort.
Specifically, FeMLoc achieves up to 80.95\% improvement in localization accuracy compared to the conventional baseline neural network (NN) approach after only 100 gradient steps. Alternatively, for a target accuracy of around 5m, FeMLoc achieves the same level of accuracy up to 82.21\% faster than the baseline NN approach.
This translates to FeMLoc requiring fewer training iterations, thereby significantly reducing fingerprint data collection and calibration efforts.
Moreover, FeMLoc exhibits enhanced scalability, making it well-suited for location-aware massive connectivity driven by emerging wireless communication technologies.
\end{abstract}

\section{Introduction}
\IEEEPARstart{T}{he} emergence of wireless technologies like 5G and beyond \cite{survey5g}, coupled with the integration of unmanned aerial vehicle-assisted networking \cite{nabiew}, is fueling a new era of ubiquitous connectivity. This, in turn, is driving the growth of location-based services (LBS) that are reshaping our daily lives and experiences.

By harnessing the power of geolocation data, LBS have revolutionized navigation, personalized recommendations, contextual information delivery, asset tracking, safety measures, and personalized services in various industries \cite{horsmanheimo2019indoor}. 

However, the potential of LBS remains constrained within enclosed spaces where traditional outdoor positioning systems, such as global positioning systems (GPS), often falter. To overcome this limitation and unlock the full capabilities of LBS within indoor environments, indoor localization technologies have emerged as a challenging area of research and development.

Indeed, indoor localization has seen significant growth due to the increasing demand for LBS supported by recent advancements in wireless technologies, and the proliferation of Internet of Things (IoT) devices \cite{Jouhari2022ASO}. Technologies like Wi-Fi, Bluetooth, inertial measurement units (IMUs), ultra-wideband (UWB), computer vision, and sensor fusion are driving this rise \cite{furfari2019next}.

With the densification of connected devices, traditional localization methods (i.e., trilateration, triangulation, fingerprinting, etc.) using wireless signal metrics (i.e., time of arrival (ToA), time difference of arrival (TDoA), radio signal strength indicator (RSSI), channel state information (CSI), etc.) fail to provide accurate localization and suffer from scalability, and the high dynamicity of indoor environments.

Hence the need for data-driven solutions where machine learning has proven to be a good fit, as enlightened in \cite{Etiabi2023FederatedLB}. 
In fact, with the uprise of machine learning (ML), RSSI \cite{Etiabi2022SpreadingFA} and CSI fingerprinting \cite{DBLP:journals/corr/abs-2210-14510} have taken over all the other wireless signal metrics used for localization. RSSI and CSI fingerprinting have become the primary tasks in designing an ML-based localization system. The massive adoption of ML techniques such as support vector machines (SVMs), random forests (RFs), and neural networks (NNs) allows indoor localization systems to achieve higher accuracy and robustness compared to traditional methods \cite{nessa2020survey}. 

However, traditional ML approaches face challenges related to data privacy, scalability, and generalization. Additionally, the collection of fingerprints in addition to being overwhelming and effort-intensive, is space- and time-bounded due to the dynamics of indoor environments. As a result, in
the literature, every single solution was designed for a specific indoor scenario and fails in scaling and generalization, promoting repetitive localization tasks with unused or underutilized knowledge from prior works or other environments. 

To overcome these limitations, we introduce FeMLoc, a federated meta-learning framework for fast adaptive localization tasks in IoT networks. This framework takes advantage of two
foundational learning paradigms namely federated learning (FL) \cite{McMahan17} and meta-learning (MTL) \cite{FinnAL17}.

Federated learning enables collaborative training of ML models across multiple edge devices or data sources while preserving data privacy. By leveraging local data, FL allows models to be trained directly on user devices, eliminating the need to transmit sensitive location data to a centralized server \cite{Etiabi2022FederatedDB, rev1p1}. 

The FL, however, typically provides a model for a specific task, which means that the resulting trained model is only beneficial for IoT devices (clients) who have participated in the training. Any new, unseen clients may be disadvantaged or discriminated against, especially if their tasks differ from those for which the FL was designed.

In response to this shortcoming of FL,
MTL emerges as a promising direction to enhance the adaptability and efficiency of indoor localization systems. Meta-learning leverages insights gained from previous learning experiences to facilitate fast adaptation to new environments or tasks.
The property becomes especially crucial in radio fingerprinting localization, where an extensive number of channel measurements are conducted to construct the radio map of the continuously sensed wireless environment \cite{hintikka, delaccess}.

MTL enables indoor localization models to acquire knowledge from a multitude of environments, leveraging this knowledge to rapidly adapt to novel indoor contexts, thereby reducing the need for extensive training data in each specific environment \cite{DBLP:journals/icl/WeiYWWZ23, DBLP:conf/icc/GaoZKYX022, DBLP:conf/globecom/Long0XH22}.
In light of the aforementioned advancements and challenges, adopting federated meta-learning for indoor localization holds tremendous potential, hence our proposed framework, FeMLoc.

In FeMLoc, MTL is employed centrally to build a meta-model that captures the underlying patterns of indoor localization across diverse environments. The meta-model is trained to be adaptable, learning from various indoor environments to quickly adjust to new, unseen environments with minimal additional data collection, thereby reducing calibration efforts. 

On the other hand, FL enables collaborative training of the meta-model (global meta-model) on the central server, utilizing local updates from edge devices without sharing their private data. This approach protects user privacy while facilitating knowledge sharing across devices. Essentially, FL fosters collaboration on data, while MTL empowers the model to learn from this collaborative knowledge.

By introducing FeMLoc, this research paper makes substantial contributions in several distinct areas, each showcasing its unique impact and benefits:
\begin{itemize}
    \item \textbf{Privacy Assurance:} FeMLoc addresses data privacy concerns by exclusively conducting model updates on individual client devices. This approach guarantees the confidentiality of sensitive location data, ensuring robust privacy protection without compromising the model's performance.
    \item \textbf{Inherent Scalability and Collaborative Training:} FeMLoc introduces inherent scalability by enabling collaborative model training across diverse clients. This emphasizes the framework's ability to scale without centralized data storage, facilitating seamless adaptation to evolving IoT network requirements.
    \item \textbf{Dynamic Adaptability for Real-World Scenarios:} The fusion of federated learning and meta-learning in FeMLoc creates a dynamic and adaptable framework. This unique combination allows rapid adaptation to changes in data distribution, evolving tasks, and varying network and environmental conditions, ensuring robust performance in real-world indoor positioning scenarios.
\end{itemize}
\noindent The remainder of this paper is organized as follows: Section \ref{sec:relworks} presents the related works while Section \ref{sec:sysm} depicts the problem formulation and the system model description. In Section \ref{sec:femloc}, we present our proposed framework whose performance evaluation is presented in Section \ref{sec:perfeval}. Finally, we conclude our work and set forth some perspectives in Section \ref{sec:conclusion}.

\section{Related works}
\label{sec:relworks}

\begin{table*}[!t]
\centering
\caption{RSSI Fingerprinting datasets}
\label{tab:datasets}
\begin{tabular}{l|l|l|l|l|l|l|l}
\hline
Reference    & Signal type & Publication date & Building & Floor & Area ($m^2$)  & \#Samples (train + test)& \#WAPs \\
\hline
UJIIndoorLoc \cite{ujiindoorloc} & Wi-Fi      & 2014    & 3        & 13    & 108703 & 21048(1997 + 1111)    & 520 \\
TUT2017  \cite{lohan_elena_simona_2017_889798} & Wi-Fi   & 2017  & 1 & 5  &2,570 (208x108)& 4648(697+3951)   & 991  \\
TUT2018 \cite{philipp_richter_2018_1161525}    & Wi-Fi     & 2018      & 1     &  3  &    --    & 1428 (446+982) & 489 \\
IEEEDataport \cite{49yg-5d21-19} & Wi-Fi & 2019 & 1 & 2 & 24,000 & 7175 & 72\\
UJILIBDB \cite{ltmsupport} & Wi-Fi &2017 & 1& 2& -- &63504 & 448\\
HDLCDataset \cite{hybrid} &Wi-Fi \& BLE & 2022 & 1 & 3 &386.25 & & 42+17 \\
UTSIndoorLoc \cite{utsindoorloc} &  Wi-Fi &2019 & 1& 16 & 44,000 &  9494 (9107 + 387) &  589  \\
UTMInDualSymFi \cite{utm} & Wi-Fi& 2022& 4& 14& & 53300& 250\\
DSIData \cite{dsi} & Wi-Fi& 2020& 1& 1& -- &1717 & 157\\
MSIData \cite{msi} & Wi-Fi& 2019 & 1& 1& 1000 (50x20) & 28915& 11\\
MTUT\_SAH1 \cite{mtut} & Wi-Fi &2021 & 1& 3& --& 9447& 775 \\
MTUT\_TIE1 \cite{mtut} & Wi-Fi &2021 & 1& 6 & -- &10683 & 613\\
\hline
\end{tabular}
\end{table*}

\subsection{RSSI Fingerprinting for Indoor Localization}
RSSI fingerprinting has been widely employed for indoor localization, leveraging the RSSI from wireless sources to determine location coordinates within buildings. Early methods such as K-nearest neighbors (KNN) and Gaussian processes have demonstrated effectiveness in mapping RSSI fingerprints to indoor positions \cite{knn, gp}. 

However, these traditional approaches often face challenges in scalability and adaptability to dynamic indoor environments, where signal strength variations and architectural changes can lead to localization inaccuracies.

Moreover, the calibration of such fingerprints is environment-dependent, leading to a tremendous release of RSSI datasets in the literature as shown in TABLE \ref{tab:datasets}. Indeed, for each RSSI-based indoor localization task, data are collected and calibrated to build a localization model for the specific environment. When an environmental or network change occurs, recalibration is needed and new data are collected. 

As a first attempt to address this concern, many deep learning techniques have been introduced to improve the localization systems by providing better accuracy and relatively low complexity \cite{zhou2022deepvip, yang2023deepwipos, liu2022uwb}. But still, these deep neural network (DNN) solutions fail to deal with the high variability in indoor spaces, further maintaining periodic re-calibrations.

To reduce this effort-intensive data collection and calibration, different learning paradigms of the DNN models have been proposed in the literature attempting to address these challenges. Notably, these include federated learning, transfer learning, and meta-learning, each with its strengths and limitations. 


\subsection{Federated Learning for Indoor Localization}
Federated learning has emerged as a promising paradigm for collaborative model training across distributed clients while preserving data privacy \cite{McMahan17}. Several studies have explored the application of federated learning to indoor localization tasks, enabling collective model updates from multiple sources without the need for centralizing their data \cite{Etiabi2022FederatedDB, FedLoc2020, Etiabi2023FederatedLB}.
This approach addresses concerns related to data security \& complexity and communication bottlenecks, showing potential for efficient indoor positioning solutions.

It facilitates the aggregation of diverse data sources, addresses privacy concerns by keeping data locally, conserves bandwidth resources by transmitting model updates instead of raw data, and promotes collaborative learning to improve the accuracy and adaptability of indoor positioning models. By harnessing the power of FL, indoor positioning systems can leverage the benefits of a large dataset, ensure privacy compliance, and overcome communication limitations, thus paving the way for more effective and privacy-preserving indoor localization solutions.

In \cite{FedLoc2020} for instance, the authors showcase a general FL framework for indoor localization where a set of devices collaboratively train a joint DNN model designed for localization problems.
Similarly \cite{Etiabi2023FederatedLB} proposed a hierarchical location prediction model built upon a federated learning framework which prove to be efficient in terms of accuracy and communication overhead compared to the traditional ML model.

Other works \cite{ibnatta2022indoor,9761235,guo2023fedpos} delve into the same direction by putting FL at the center of their localization systems design.
But still, FL applied to RSSI-based indoor positioning encounters limitations that include the heterogeneity of RSSI data collected from different devices, scalability challenges with a growing number of federated nodes, and difficulties in generalizing FL models across diverse indoor environments. 

Indeed, FL, in its typical implementation, produces task-specific models that are beneficial only to the IoT devices (clients) that participated in the training process \cite{rev1p2}. This limitation can pose challenges for new and unseen clients, as they may face disadvantages or discrimination if their tasks deviate from those for which the FL model was originally designed. The lack of adaptability and generalization of FL models to accommodate diverse tasks and scenarios hinders the equitable and inclusive application of FL in dynamic and evolving environments, hence the need to consider new paradigms including transfer learning and meta-learning.
\subsection{Transfer Learning in Indoor Localization}
The field of indoor localization has witnessed considerable exploration of transfer learning techniques to enhance localization performance\cite{chen2022few}. Previous research has delved into leveraging knowledge gained from one indoor environment to improve localization accuracy in a different setting. This often involves transferring pre-trained models or features from a source domain (the original indoor environment) to a target domain (the new indoor environment), with the aim of mitigating the need for extensive data collection and training in the new environment \cite{maghdid2022enabling, guo2023fedpos, stahlke2022transfer}.

In these works, existing models were originally constructed in typically well-mapped source environments using data primarily obtained from these environments. These models subsequently serve as foundational frameworks for adaptation to target environments. The adaptation process involves fine-tuning the parameters of the new indoor localization models using data samples collected from the target environments.
This approach has exhibited the potential to enhance the performance of localization models within these designated environments.

However, transfer learning typically relies on a single source environment, which may not capture the full range of variability and complexities present in diverse indoor environments. This constraint hinders the effectiveness and the generalization of learned models to various unseen target environments.

Another shortcoming of conventional transfer learning for indoor localization is its assumption that source and target environments share the same characteristics, especially the network layouts which that not reflect the reality of the actual indoor environment heterogeneity.
These limitations have prompted the emergence of meta-learning techniques in indoor localization.

\subsection{Meta-Learning in Indoor Localization}
Meta-learning, also known as "learning to learn," focuses on enabling models to quickly adapt and generalize to new environments or tasks by leveraging knowledge acquired from previous learning experiences, as shown in \cite{FinnAL17} where its effectiveness has been established.
Within the field of indoor localization, recent studies have explored the potential of meta-learning techniques to enhance the accuracy, robustness, and efficiency of indoor positioning systems \cite{DBLP:journals/corr/abs-2210-14510, DBLP:journals/icl/WeiYWWZ23, owfi2023meta}. 

 Indeed, in \cite{DBLP:journals/corr/abs-2210-14510} for instance, the authors proposed a meta-learning-based localization system with CSI fingerprinting within a multi-environment-based mobile network. Their solution consists of a deep learning (DL) model with two components: the first part focuses on acquiring environment-independent features, while the second part integrates these features based on the specific environmental context. Their contribution lies in introducing a meta-learning approach to training the first part across multiple environments in order to improve the model's generalizability.
 
 Subsequently, in \cite{DBLP:conf/icc/GaoZKYX022} and \cite{metalocj} the authors developed a meta-learning-based indoor localization system where the model that is assumed to be a DNN, and they optimized a set of group-specific meta-parameters using historical data collected from diverse, well-calibrated indoor scenarios and the maximum mean discrepancy criterion.
As of the time of this writing, the most recent contributions in that direction include \cite{DBLP:journals/icl/WeiYWWZ23} and \cite{owfi2023meta} which are both incremental to \cite{DBLP:conf/icc/GaoZKYX022}.

This glance at the literature review shows that the use of meta-learning techniques in the field of indoor localization, as shown in the above works, enables models to rapidly adapt to new indoor environments, improve localization accuracy, and adapt to changing conditions.

However, these state-of-the-art works operate under the assumption of fixed and predefined network geometries, a premise that poses challenges to the transferability of the meta-model to novel networks characterized by distinct geometries or layouts. This inherent constraint limits the adaptability of the meta-model to network environments featuring varying spatial arrangements and configurations. 

Furthermore, the training process within the presented frameworks is tailored for a predetermined and fixed number of access points. This design choice inadvertently imposes a requirement for new environments to conform to the same number of APs, restricting the versatility of the framework in accommodating diverse AP densities or network structures. 

Consequently, the applicability of their frameworks to scenarios characterized by dynamic network topologies or varying numbers of APs is notably constrained. These constraints constitute the core motivation of the new meta-learning framework we proposed in this work. With respect to data privacy and communication constraints, we built our system upon the federated learning framework, unlocking a step forward within the indoor localization domain: federated meta-learning.

\subsection{Federated Meta-Learning}
Building upon the strengths of federated learning, transfer learning, and meta-learning, federated meta-learning \cite{DBLP:journals/corr/abs-2210-13111} presents a novel approach that holds the potential to address the challenges of fingerprinting-based indoor localization in dynamic environments.
By combining collaborative model updates from diverse clients with meta-learning's rapid adaptation capabilities, federated meta-learning showcases promising potential to revolutionize indoor localization technologies.

This method has the potential to overcome challenges related to data privacy, scalability, and generalization while achieving improved localization accuracy and adaptability. Therefore, the integration of federated learning with diverse sensor data sources and meta-learning techniques opens up new avenues for adaptive and privacy-preserving indoor localization solutions.

In summary, the related works presented in this section underscore the evolving landscape of indoor localization research, showcasing the progression from traditional methods like RSSI fingerprinting and transfer learning to emerging paradigms such as federated learning and meta-learning. Building upon these foundations, the introduction of federated meta-learning \cite{fedmeta} offers a compelling avenue to enhance indoor positioning solutions, by leveraging collaborative updates, rapid adaptation, and knowledge transfer across diverse indoor environments. The subsequent sections delve deeper into our proposed federated meta-learning framework, highlighting its unique attributes and contributions in addressing the challenges of fingerprinting-based indoor localization in dynamic settings.
\section{System model and problem formulation}
\label{sec:sysm}
\subsection{Problem Formulation}
Our research is focused on developing an indoor positioning system for IoT devices that utilizes fingerprinting techniques. Traditionally, creating a radio map of the indoor wireless environment requires collecting radio fingerprints at various reference points, which is a time-consuming and labor-intensive process, especially when high precision is necessary. Previous works on indoor localization have followed this approach for different indoor scenarios and environments.

Our objective is to break free from the cycle of continuous fingerprint collection for new systems. Instead, we aim to use the knowledge gained from previous research to simplify the fingerprint collection process and expedite the development of new positioning systems. By distilling insights and experiences from prior works, we aim to reduce the complexity associated with fingerprint collection and accelerate the implementation of new positioning systems formulated as localization tasks.
\begin{definition}[Localization task] 
\label{def:task}
Given an area of interest in a wireless environment, a localization task involves learning to map radio fingerprint measurements from the available APs to the locations where the measurements were taken.
Such a task is defined by $\mathcal{T}(\mathcal{D},\mathcal{L})$, with $\mathcal{D}$ the radio fingerprints dataset and $\mathcal{L}$ the loss associated with a learning model.
The set of available APs constitutes the signal space of the localization task.
Furthermore, we designate by $\rho(\mathcal{T})$ the distribution of the localization tasks $\mathcal{T}$. Indeed, this distribution provides valuable insights into the diversity and characteristics of the localization tasks within our study. 
\end{definition}

As for conventional ML, the dataset is split into training and test sets corresponding to the support set $\mathcal{D}^s$ and query set $\mathcal{D}^q)$ respectively in MTL, that is $\mathcal{D} = (\mathcal{D}^s, \mathcal{D}^q)$.
In the case of indoor localization, each localization task associated with a dataset $\mathcal{D}$ is specifically designed for a particular indoor space, resulting in repetitive indoor positioning system (IPS) design efforts. This repetition is due to the high variability in wireless infrastructure, geometry, occupancy, and other configurations across different indoor environments. Unfortunately, existing IPS designs do not take advantage of prior works, leading to a significant loss in human effort.

To address this issue, our work aims to leverage federated meta-learning to develop a framework that accelerates the development of new localization systems by extracting knowledge from previous works. We assume the existence of a set of localization tasks that will be used to train the proposed framework. However, it is crucial to note that these tasks involve different input spaces and network geometries depending on the specific indoor environment. This diversity makes it impractical to utilize the conventional MTL framework described in the literature.
Therefore, our proposed approach presents a tailored solution that is suitable for dynamic indoor localization tasks, overcoming the limitations of existing MTL frameworks.
\begin{definition}[Meta signal space] 
\label{def:msignal}
Suppose that wireless environments characterizing the tasks $\{\mathcal{T}_k\}_{k=1,2,...,K}$ have different signal spaces, i.e., different scales in terms of the number of APs as well as the deployment strategies. For the training tasks, the median of the respective numbers of APs is considered as the meta signal space dimension and it is given by:
\vspace{-1.5mm}
\begin{equation}
\label{eq:signal}
d = \operatorname{Median}(\mathcal{M}) 
= \begin{cases}\mathcal{M}\left[\frac{K+1}{2}\right] & \text { if } \mathrm{K} \text { is odd } \\
\frac{\mathcal{M}\left[\frac{K}{2}\right]+\mathcal{M}\left[\frac{K}{2}+1\right]}{2} & \text { if } \mathrm{K} \text { is even }\end{cases},
\vspace{-1.5mm}
\end{equation}
where $\mathcal{M} = \{m_1, m_2, ..., m_K \}$ is the ordered list of the numbers of APs and
$K$ the number of localization tasks. An environment is down-scaled when the original signal space is higher than the meta signal space and up-scaled otherwise to match the set meta signal space.
\end{definition}
To balance the up-scaling and down-scaling of training environments, we choose to use the median instead of the mean. The mean tends to focus on the overall available access points. As will be described in Section \ref{sec:sysm}, our proposed framework introduces the meta-signal space, which allows the framework to adapt to the signal space of each training task. This yields suitable and efficient multi-task learning, where meta-knowledge is learned about the existing tasks so that any unseen tasks can rapidly adapt with just a few measurement samples.
\begin{definition}[Meta-knowledge $\boldsymbol{\theta}$]
By  ``learning how to learn'', the MTL algorithm ends up with, to some extent, a deep understanding of the underlying properties of the existing learning tasks, embedded in a model $\boldsymbol{\theta}$ called meta-knowledge. This is done through meta-training over the existing learning tasks, empowering $\boldsymbol{\theta}$ with the ability to generalize and quickly adapt to new unseen tasks with similar distribution during meta-testing.
\end{definition}

In our investigation of indoor localization, we acknowledge the significance of privacy concerns that may hinder access to certain location data. Consequently, we adopt a decentralized approach by avoiding the centralization of all localization tasks on a server. Instead, we seek to foster client collaboration through the implementation of federated learning. Given the diverse geometries present in indoor spaces across various tasks, it becomes imperative to develop a personalized learner for each specific task. Our primary objective is to devise a multi-task learning framework that leverages personalized federated learning to generate meta-knowledge for indoor localization tasks. This framework aims to facilitate rapid adaptation to new system scenarios while minimizing the fingerprint calibration efforts required.

\subsection{System Model: General Architecture}
\begin{figure}[!t]
    \centering
    \includegraphics[width=\linewidth]{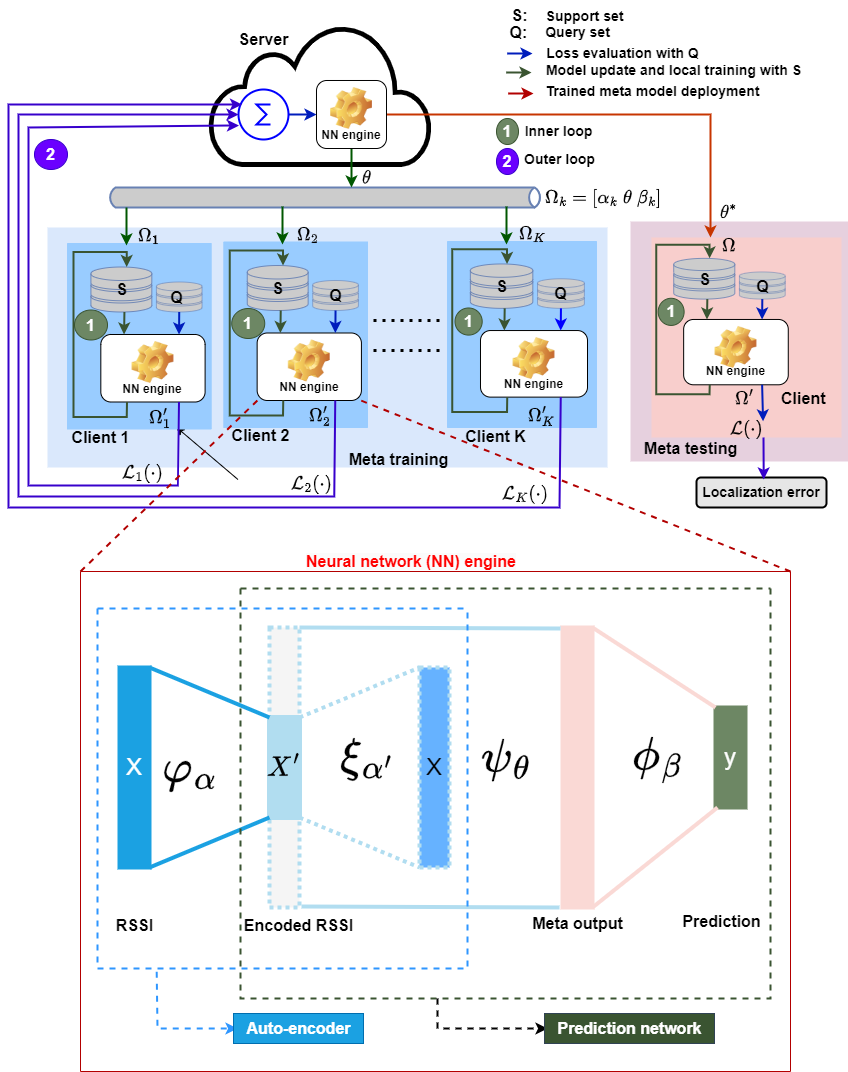}
    \caption{Proposed FeMLoc architecture.}
    \label{fig:femloc}
    \vspace{-0.2in}
\end{figure}
In this section, we elucidate the comprehensive framework depicted in Fig. \ref{fig:femloc}.
The overall framework for the localization task consists of three key components: The auto-encoder, the meta-learner, and the locations mapper. It is important to address the diverse signal spaces present in localization tasks to effectively develop a shared meta-learner. By defining the meta-learning signal space using Eq. (1), two distinct task groups are identified based on their signal space dimensions.

One group exhibits a signal space dimension smaller than that of the meta-learner, while the other group has a higher dimensionality. Instead of handling these groups of learning tasks separately, we propose employing a global solver that automatically adjusts the signal space dimension to match the meta-signal space. To achieve this, we introduce an auto-encoder responsible for extracting input features from each learning task and converting them into a latent representation that aligns with the meta-model's input features space.
This is accomplished by the encoder, which learns the latent representation of the input, and the decoder, which reconstructs the input from the latent representation, ensuring a functional encoding of the input features. Specifically, given the input signal X, the auto-encoder process is as follows:
\begin{equation}
\begin{matrix}
\begin{aligned}
\varphi: \;\mathbb{R}^m &\to \mathbb{R}^d  \\ 
 \; \mathbf X &\mapsto \mathbf X^\prime = \varphi_{\boldsymbol{\alpha} } \left ( \mathbf X \right ) 
\end{aligned} 
\; \text{ ; }\;
\begin{aligned}
\xi: \;\mathbb{R}^d &\to \mathbb{R}^m  \\ 
 \; \mathbf X^\prime &\mapsto \mathbf X = \xi_{\boldsymbol{\alpha}^\prime } \left ( \mathbf X^\prime \right ),
\end{aligned}
\end{matrix}
\end{equation}
where $\mathbf X^\prime$ and $\mathbf X$ are respectively the features in latent space of dimension $d$ and the input features of dimension $m$. $\boldsymbol{\alpha}$ and $\boldsymbol{\alpha}^\prime$ are the weights of the neural networks associated with the encoding and the decoding functions $\varphi(\cdot)$ and $\xi(\cdot)$, respectively.
Note that while the decoder can be utilized for data augmentation purposes, the primary objective of the auto-encoder in this context is feature extraction to adjust dimensionality and incidentally denoise the input signals. It is important to highlight that each client establishes its own auto-encoder, and the only restriction is that the latent representation must align with the meta-input space's dimension. Specifically, the encoder's output, denoted as $\mathbf X^\prime$, is used to provide input to the meta-model, denoted by $\phi$ and parameterized by $\boldsymbol{\theta}$, which aims to learn predicting a structured combination of features $\mathbf Z$ with a dimension of $n$. Subsequently, this predicted combination $\mathbf Z$ is used by the environment-specific locations mapper to predict the target coordinates. The representation of the meta-learner is thus as follows: 
\begin{equation}
    \begin{aligned}
\phi: \;\mathbb{R}^d &\to \mathbb{R}^n  \\ 
 \; \mathbf X^\prime &\mapsto \mathbf Z = \phi_{\boldsymbol{\theta} } \left ( \mathbf X^\prime \right ) =   \phi_{\boldsymbol{\theta} } \left ( \varphi_{\boldsymbol{\alpha}}\left ( \mathbf X \right ) \right ),    
\end{aligned}
\end{equation}

The essence of the framework lies in this part, wherein it undergoes joint training involving all participants. The objective is to enable new clients to utilize this jointly trained model as an initialization step, which provides crucial features already aligned with the specific environment of the client. The training procedure for this model is elaborated in the subsequent subsection.

Following this, as mentioned earlier, due to variations in the geometry of input spaces, each participant is required to develop its individual mapper to associate the features $\mathbf Z$ with the corresponding locations in its environment. The personalized location mapper $\psi$, which is parameterized by $\boldsymbol{\beta}$, is thus formally defined as follows:
\begin{equation}
    \begin{aligned}
\psi: \;\mathbb{R}^n &\to \mathbb{R}^p  \\ 
 \mathbf Z &\mapsto \mathbf y = \psi_{{\boldsymbol{\beta}}}\left ( \mathbf Z \right ) =  \psi_{\boldsymbol{\beta}}\left ( \phi_{\boldsymbol{\theta} } \left ( \varphi_{\boldsymbol{\alpha}}\left ( \mathbf X \right ) \right ) \right ), 
\end{aligned}
\end{equation}
The variable  $\mathbf y$ represents the resultant output corresponding to the p-dimensional coordinates of the target IoT device. When provided with an input $\mathbf X$, the corresponding position label is estimated through the composition of preceding functions and parameterized by aggregating the parameters from different components, which can be expressed as follows:

 \begin{equation}
\begin{aligned}
f: \;\mathbb{R}^m &\to \mathbb{R}^p  \\ 
 \mathbf X &\mapsto \mathbf{\hat{y}} = f_{\boldsymbol{\Omega}}\left ( \mathbf X \right ) =  \left (\psi_{\boldsymbol{\beta}} \circ \phi_{\boldsymbol{\theta}} \circ \varphi_{\boldsymbol{\alpha}} \right )(\mathbf X),
\end{aligned}
 \end{equation}
 with $f$ the global network function parameterized by $\boldsymbol{\Omega}  =  \left [\boldsymbol{\alpha} \;  \boldsymbol{\theta} \; \boldsymbol{\beta}   \right ]$.
Using the loss function $\ell(\cdot)$, the estimation error of a forward pass over data batch $\mathcal B$ can be expressed as follows:
\begin{equation}
    \mathcal{L}\left( f_{\boldsymbol{\Omega}} , \mathcal B\right)= \sum_{(\mathbf X,\mathbf y)\in \mathcal B} \ell \left(f_{\boldsymbol{\Omega}}(\mathbf X), \mathbf y\right).
\end{equation}

\subsection{Data Preprocessing}
\label{sec:preproc}
In this section, we detail the preprocessing steps applied to the RSSI fingerprinting databases used in different localization tasks, aiming to enhance data quality and prepare them for effective model training. In this work, we implement a set of preprocessing steps that address  common challenges in RSSI fingerprinting data. These steps include:
\begin{itemize}
    \item \textbf{APs selection}: We identify and remove APs with entirely missing RSSI values across all measurements, eliminating features with no informative data and reducing fingerprint dimensionality. Optionally, to further reduce dimensionality, we apply a visibility threshold ($\tau$). APs missing values in more than a certain percentage (e.g., $1-\tau$) of measurements are considered to have minimal contribution and are removed.
    \item \textbf{Missing value imputation}:  Missing RSSI values are imputed with a pre-defined value. We commonly replace them with the minimum observed RSSI value in the dataset, potentially adjusted by a small offset (e.g., minimum value - 1 dBm), mitigating the impact of missing data while acknowledging the absence of a signal.
    \item \textbf{Powed transformation}:
    RSSI measurements are typically negative values representing received signal strength. To improve model performance, we transform the data into a positive and normalized representation. To achieve this, we adopt the technique detailed in \cite{powed}, where RSSI measurements are converted to a powered representation defined as follows:
    \begin{equation}
        powed(RSSI_i) = \left(\frac{(RSSI_i - \text{min\_RSSI})}{\text{max\_RSSI} -\text{min\_RSSI}}\right)^\beta
    \end{equation}
    Here, $min_{RSSI}$ and $max_{RSSI}$ represent respectively the minimum and maximum observed RSSI values in the dataset, and $\beta$ is typically set to a constant value like the mathematical constant e. This transformation normalizes the data between 0 and 1 while emphasizing stronger signal strengths.
\end{itemize}
The rationale behind these techniques is to improve data quality and relevance, while opportunistic advanced feature selection is performed through dimensionality reduction in the encoder part of the general model architecture.

Note that the specific choices for missing value imputation and transformation parameters (e.g., offset value, $\beta$) might require adjustments based on the characteristics of the specific dataset for each localization task. 

\section{FeMLoc: Federated meta-learning localization model}
\label{sec:femloc}
In this section, we elucidate the overarching framework of FeMLoc, a meta-learning-based approach for indoor positioning, and outline a multitude of distinct use cases wherein FeMLoc finds application.

\begin{definition}[federated meta-learning problem]\;\\
    \begin{enumerate}
    \vspace{-0.1in}
        \item Let $\mathbb{T}$  represent the potentially infinite set of all possible data elements, such as RSSI-location pairs, where each data element is associated with a localization task $\mathcal{T}(\mathcal{D},\mathcal{L})$ as defined in Def. \ref{def:task}. For any given parameters $\boldsymbol{\Omega}$ in the parameter space $\boldsymbol\Xi$, we denote the loss at data element $\mathcal{T} \in \mathbb{T}$ with parameter $\boldsymbol{\Omega}$ as $\mathcal{L}\left( f_{\boldsymbol{\Omega}} , \mathcal D\right)$.
        \item Consider the set $\mathcal{K}$, which encompasses all possible clients. Each client denoted as $k$, $k \in \mathcal{K}$ is associated with a task distribution $\mathcal{T}$, which is supported on the domain $\boldsymbol\Xi$, such that $$k \mapsto \mathcal{T}_k = \mathcal{T}\left (\mathcal{D}_k, \mathcal{L}_k \right )$$
        \item Additionally, we assume that a meta-distribution denoted as $\mathcal{P}$  is present and operates on the set of clients $\mathcal{K}$. Each client $k$ is associated with a probability distribution $\mathcal{P}_k \in \mathcal{P}$.
    \end{enumerate}
    The objective is to optimize the expected loss at a two-level structure in the following manner:
    Then, our federated meta-learning objective can be defined as follows:
    \begin{align}
        \min_{\boldsymbol{\theta} \in \mathbb{R}^{W\in \boldsymbol\Xi}} J\left(\boldsymbol{\theta}\right):=&\mathbb{E}_{k\sim \mathcal{P}_k}\left [ \mathcal{L}_k\left(f_{\boldsymbol{\Omega}_k^{N_k}}, \mathcal{D}_k^q\right) \right] \label{eq:objective}\\
        \text{subject to: }
        &\boldsymbol{\Omega}_k^{n+1} = \boldsymbol{\Omega}_k^{n} - \mu_k \nabla_{\boldsymbol{\Omega}_k^{n}} \mathcal{L}_k\left(f_{\boldsymbol{\Omega}_k^{n}}, \mathcal{D}_k^s\right), \label{eq:innerloop}\\
        & n=1,2,3,...,N_k \nonumber\\
        &\boldsymbol{\Omega}_k^n = \left[\boldsymbol{\alpha}_k^n\; \boldsymbol{\theta}_k^n \; \boldsymbol{\beta}_k^n \right]\nonumber
\end{align}   
where $J_{\boldsymbol{\theta}}$ is the global loss of the meta learner parameterized by $\boldsymbol{\theta}$ in meta parameter space $W\in\boldsymbol\Xi$. $\mu_k$, $\mathcal{D}_k^s$, and $\mathcal{D}_k^q$ are the inner learning rate, the support, and query sets of client $k$, respectively. 

\end{definition}


The federated meta-objective in \eqref{eq:objective} can be broken down into two-level optimization corresponding to the two levels of training in the FL setting: the inner-loop optimization and the outer-loop optimization
\subsection{Inner loop optimization: FL local training}
In this level, each client $k$  trains the meta-model $\boldsymbol{\theta}$ which is part of its personalized model $\boldsymbol{\Omega}$.
The inner loop optimization is done through iterative weight updates using the stochastic gradient descent (SGD) method \cite{yuan2020federated} as shown in \eqref{eq:innerloop}.

With every client $k$ in the FL loop, we have the following:
\begin{equation}
    \boldsymbol{\Omega}_k^{n,r} =  \left(\boldsymbol{\alpha}_k^{n+r\times N_k}, \boldsymbol{\theta}_k^{n,r}, \boldsymbol{\beta}_k^{n+r\times N_k} \right),
\end{equation}
where $N_k$ is the total number of local updates of client $k$, $r$ is the current FL iteration (communication round), and $\boldsymbol{\theta}_k^{n,r}$ is the current local update (the $n^{th}$) of the meta-model received from the server at the $r^{th}$ communication round.
For the sake of simplicity, let $\boldsymbol{\theta}_k^{0,r} =\boldsymbol{\theta}_k^{r}=\boldsymbol{\theta}^{r}$ and $\boldsymbol{\theta}_k^{N,r} =\boldsymbol{\theta}_k^{N}$.
Additionally, the probability distribution of client $k \sim \mathcal{P}_k \in \mathcal{P}$ can be simplified by
$\rho_k = \frac{\left|\mathcal{D}_k^q\right|}{\sum_{k^\prime} \left|\mathcal{D}_{k^\prime}^q\right|}$ is the contribution factor of the client $k$ to global update of the meta-model.

At each communication round $r$, client $k$ receives the global meta-model $\boldsymbol{\theta}_k^{r}$ from the server and uses it to initialize its local training. In contrast, when initializing the local model at communication round $r$, the encoder $\boldsymbol{\alpha}_k$ and location mapper $\boldsymbol{\beta}_k$ resume their training from the previous FL round.
\subsection{Outer loop optimization: FL global update}
This is where the parameters of the shared meta-model are globally updated. Indeed, after the inner loop optimization, i.e., the local training, each client $k$ evaluates the gradients of the locally updated meta-model $\boldsymbol{\theta}_k^N$ and sends them to the server. That is after $N$ local steps of SGD, client $k$ transmits the gradients with respect to $\boldsymbol{\theta}_k^{N}$ 
to the server. 
Upon reception of the gradients from all the clients, the FL server updates the global meta-model as follows:
\begin{equation}
    \boldsymbol{\theta}^{r+1} = \boldsymbol{\theta}^r - \eta \sum_k \rho_k \nabla_{\boldsymbol{\theta}_k^N } \mathcal{L}_k(f_{\boldsymbol{\Omega}_k^{N,r}}, \mathcal{D}_k^q),
    \label{eq:outerloop}
\end{equation}
where $\eta$ is the outer learning rate.

We denote by $\boldsymbol{\theta}^{*} $ the optimal solution for \eqref{eq:objective}. The goal is to approximate this solution after $R$ communication rounds in the federated setting described in the next section, such that $\boldsymbol{\theta}^{R} \sim \boldsymbol{\theta}^{*} $. 


\subsection{Federated meta-learning algorithm}
Our federated meta-learning framework encompasses two distinct phases: the meta-training phase and the meta-testing phase
\SetKwComment{Comment}{/* }{ */}
\RestyleAlgo{ruled}
\begin{algorithm}
\SetAlgoLined
  \caption{FeMLoc meta-training}
  \label{alg:femloc-train}
  \SetKwInOut{Input}{Inputs}
  \SetKwInOut{Output}{Outputs}
  \Input{$\{\mathcal{D}_k\}$ \Comment*[l]{Clients' RSSI Datasets}}
  \Input{$\{\mathcal{L}_k\}$ \Comment*[l]{Clients' loss functions}}
  
  \Output{$\boldsymbol{\theta}^*:=\boldsymbol{\theta}^R$ \Comment*[l]{Trained meta-model}}
  \SetKwProg{ServerInit}{ServerInit}{}{}
   \ServerInit{()}{
   $ r \gets 0$ , set $\boldsymbol{\theta}^r$, $\boldsymbol{\alpha}_k$, and $\boldsymbol{\beta}_k$\;
  Clients' models' configuration: $\boldsymbol{\Omega}_k \gets \left(\boldsymbol{\alpha}_k, \boldsymbol{\theta}^0, \boldsymbol{\beta}_k\right)$\;
  
  }
  \While{not converged and $r < R$}{
  server broadcasts $\boldsymbol{\theta}^r$\;
  
  \ForEach{Client $k \in \{1,2, \dots,K\}$ }{
  \Comment{In parallel}
  \SetKwProg{ClientLocalTraining}{ClientLocalTraining}{}{}
   \ClientLocalTraining{$(\mathcal{D}_k, \mathcal{L}_k, \boldsymbol{\theta}^r$)}{
    $ n \gets 0$,  $\boldsymbol{\theta}_k^{n,r} \gets \boldsymbol{\theta}^r$\;
    
    Update the local model with: 
    $\boldsymbol{\Omega}_k^{n,r} \gets \left(\boldsymbol{\alpha}_k^{n+r\times N_k} , \boldsymbol{\theta}_k^{n,r} , \boldsymbol{\beta}_k^{n+r\times N_k} \right)$\;
    \For{$n \gets 1$ \KwTo $N$}{
    Train the local model using \eqref{eq:innerloop}\;
    }
    Get trained model:
    $\boldsymbol{\Omega}_k^{N_k,r} = \left(\boldsymbol{\alpha}_k^{N_k+r\times N_k} , \boldsymbol{\theta}_k^{N_k,r} , \boldsymbol{\beta}_k^{N_k+r\times N_k} \right)$\;
     Client uploads $\boldsymbol{\theta}_k^{N_k}=\boldsymbol{\theta}_k^{N_k,r}$\;
  }
  }
  
  \SetKwProg{ServerGlobalUpdate}{ServerGlobalUpdate}{}{}
  \ServerGlobalUpdate{$(\{\boldsymbol{\theta}_k^{N_k}\})$}{
  Update the global meta-model $\boldsymbol{\theta}^r$ using \eqref{eq:outerloop}\;
  }
  $ r \gets r +1$ \;
  
  }
  \end{algorithm}
The optimization problem defined in \eqref{eq:objective}, which can be reduced to solving the problem defined by \eqref{eq:outerloop}, is resolved using a collaborative approach by different clients $\{k\}$, each of whom is assigned a localization task $\mathcal{T}_k$. To do this, the resolution through the combination of FL and meta-learning can be summarized in the following steps, outlined in Algorithm \ref{alg:femloc-train}.

\subsubsection{Meta model initialization}
This step involves randomly initializing a neural network on the server, which is then distributed to all participating clients responsible for localization tasks.
\subsubsection{Local update} At the beginning of each communication round, each client uses a local copy of the meta-model to augment its model architecture, constituting an auto-encoder and a private locations mapper, thereby forming the general model as described in Section \ref{sec:sysm}.
Locally, the meta-training process consists of two phases: (i) Training, where the client's support set is used to train the model with a few gradient steps, and (ii) Query (Evaluation), where the client's query set is used to evaluate the model, and the gradients of the corresponding loss are transmitted to the server.

\subsubsection{Global update}
This step entails the central server's role in the aggregation of gradients contributed by all participating clients, thereby deriving the global gradients. Subsequently, these global gradients are employed to update the global meta-model. The updated global meta-model is then disseminated to all client nodes for use in subsequent training rounds, thereby facilitating an iterative process that continues until convergence is achieved.

Elaborating on the previously outlined federated meta-training procedure, the subsequent phase involves what is referred to as ``meta-testing.'' Within this context, the pre-trained meta-model is deployed into novel environments, constituting new tasks or clients.  Given a new environment, a new localization is assigned to client \mbox{$\kappa\sim\mathcal{T}_\kappa=\mathcal{T}\left(\mathcal{D}_\kappa, \mathcal{L}_\kappa\right)$} whose goal is to exploit the pretrained meta-model to design a new personalized model. To do so, the new client initializes its model with the pretrained weights of the meta-model and derives its task-specific parameters by taking $N_\kappa$ gradient steps ( $N_\kappa-$shot learning), resulting in the following:
\begin{equation}
    \begin{aligned} 
        \boldsymbol{\Omega}_\kappa^{N_\kappa} &=  \boldsymbol{\Omega}_\kappa^{m} - \mu_\kappa \Biggl[\nabla_{\boldsymbol{\Omega}_\kappa^{m}} \mathcal{L}_\kappa\left(f_{\boldsymbol{\Omega}_\kappa^{m}}, \mathcal{D}_\kappa^s\right)+ \\
         &\sum_{i=1}^{N_\kappa-1} \nabla_{\boldsymbol{\Omega}_\kappa^{i}} \mathcal{L}_\kappa\left(f_{\boldsymbol{\Omega}_\kappa^{i}}, \mathcal{D}_\kappa^s\right) \Biggl] 
    \end{aligned},
    \label{eq:fewshot}
\end{equation}
with \mbox{$\boldsymbol{\Omega}_\kappa^{m}  = \left ( \boldsymbol{\alpha}_\kappa^{0}, \boldsymbol{\theta}^{*}, \boldsymbol{\beta}_\kappa^{0} \right )$}, and the optimal task-specific model is obtained through the general optimization problem below.
\begin{equation}
    \boldsymbol{\Omega}_\kappa^{*}  = \argmin_{\boldsymbol{\Omega}_\kappa}\mathcal{L}_\kappa\left(f_{\boldsymbol{\Omega}_\kappa}, \mathcal{D}_\kappa^s\right),
    \label{eq:opti}
\end{equation}
and the final evaluation loss for each test client $\kappa$ is given by
\begin{equation}
    \mathcal{L}_{\kappa}^* = \mathbb{E}_{(\boldsymbol X, \boldsymbol y) \sim \mathcal{D}_\kappa^q}\ell \left(f_{\boldsymbol{\Omega}_\kappa^*}(\mathbf X), \mathbf y\right).
\end{equation}
The primary motivation underpinning our framework is the expectation that, through this approach, $ \boldsymbol{\Omega}_\kappa^{N_\kappa}$ can be a good approximation of $\boldsymbol{\Omega}_\kappa^{*} $ i.e., \mbox{$ \boldsymbol{\Omega}_\kappa^{N_\kappa} \sim \boldsymbol{\Omega}_\kappa^{*} $} with $N_\kappa$ as small as possible, demonstrating an expedited convergence to satisfactory performance. This deployment of the meta-model into a new environment and the computation of $\boldsymbol{\Omega}_\kappa^{N_\kappa}$ is elucidated in Algorithm \ref{alg:femloc-test}, which serves as the foundational framework guiding this meta-testing process.
\SetKwComment{Comment}{/* }{ */}
\RestyleAlgo{ruled}
\begin{algorithm}
\SetAlgoLined
  \caption{FeMLoc meta-testing}
  \label{alg:femloc-test}
  \SetKwInOut{Input}{Inputs}
  \SetKwInOut{Output}{Outputs}
  \Input{$\{\mathcal{D}_\kappa\}$ \Comment*[l]{Test Clients' RSSI Datasets}}
  \Input{$\{\mathcal{L}_\kappa\}$ \Comment*[l]{Test Clients' loss functions}}
  \Input{$\boldsymbol{\theta}^*$ \Comment*[l]{Global meta-model downloaded from the server}}
  
  \Output{$\{\boldsymbol{\Omega}_\kappa^*:=\boldsymbol{\Omega}_\kappa^{N_\kappa}\}$ \Comment*[l]{Trained model}}

  \ForEach{Client $\kappa \in$ TestClients }{
  
  \SetKwProg{ClientInit}{ClientInit}{}{}
   \ClientInit{()}{
   $ n \gets 0$,  $\boldsymbol{\theta}_\kappa^{n} \gets \boldsymbol{\theta}^*$, set $\boldsymbol{\alpha}_\kappa^{n}$ and $\boldsymbol{\beta}_\kappa^{n}$\;
   Client model's configuration: 
    $\boldsymbol{\Omega}_\kappa^{n} \gets \left(\boldsymbol{\alpha}_\kappa^{n} , \boldsymbol{\theta}_\kappa^{n} , \boldsymbol{\beta}_\kappa^{n} \right)$\;
  }
  
  \SetKwProg{ClientTraining}{ClientTraining}{}{}
   \ClientTraining{$(\mathcal{D}_\kappa, \mathcal{L}_\kappa, \boldsymbol{\theta}^*$)}{
    \For{$n \gets 1$ \KwTo $N_\kappa$}{
    Train the client model using \eqref{eq:innerloop}\;
    }
    Get trained model $\boldsymbol{\Omega}_\kappa^{N_k} = \left(\boldsymbol{\alpha}_\kappa^{N_\kappa} , \boldsymbol{\theta}_\kappa^{N_\kappa} , \boldsymbol{\beta}_\kappa^{N_\kappa} \right)$\;
  }
  }
  \end{algorithm}

\subsection{Theoretical convergence and transferability analysis}
In the following section, we delve into a critical aspect of our study. This phase involves an exploration of the underlying theoretical foundations that govern the convergence properties of our federated meta-learning framework. Moreover, we investigate the transferability of knowledge acquired by the pre-trained meta-model when applied to entirely new and uncharted domains. This examination is fundamental in establishing the robustness and adaptability of our approach, shedding light on its potential implications for a wide array of real-world applications.
\begin{definition}
  The meta-model $\boldsymbol{\theta}^{r}$ at the $r^{th}$ communication round is said to be an $\varepsilon-$accurate solution for the federated meta-learning problem reduced to the optimization problem in \eqref{eq:outerloop} if the following criterion is fulfilled.
    \begin{equation*}
       \begin{aligned}
           &\mathbb{E}_k\bigg[\left \| \nabla_{\boldsymbol{\theta}^{r}} \mathcal{L}_k\left(f_{\boldsymbol{\Omega}_k^{N_k,r}}, \mathcal{D}_k^q\right) \right \|^2\bigg]< \varepsilon, & \\
         &\boldsymbol{\Omega}_k^{N_k,r} = \left(\boldsymbol{\alpha}_k^{N_k(r+1)}, \boldsymbol{\theta}^{r}, \boldsymbol{\beta}_k^{N_k(r+1)} \right).
       \end{aligned}
    \end{equation*}
    \label{def:accuracy}
\end{definition}
Furthermore, we make the following assumptions.
\begin{assumption}[Smoothness]\;\\
    a) For each client $k$ with task \mbox{$\mathcal{T}_k = \mathcal{T}\left (\mathcal{D}_k, \mathcal{L}_k \right )$}, the gradients of the loss function $\mathcal{L}_k \left (\cdot,\cdot\right )$ are assumed to be \mbox{$\delta_1$-smooth}, so that for any \mbox{$\boldsymbol{\Omega}_k,\boldsymbol{\Omega}_k^{\prime} \in \boldsymbol\Xi$} we have
    \begin{equation*}
    \begin{aligned}
        &\left \| \nabla_{\boldsymbol{\Omega}_k} \mathcal{L}_k\left(f_{\boldsymbol{\Omega}_k^\prime}, \mathcal{D}_k^\sigma \right) - \nabla_{\boldsymbol{\Omega}_k} \mathcal{L}_k\left(f_{\boldsymbol{\Omega}_k^{\prime}}, \mathcal{D}_k^\sigma \right) \right \| \leq \delta_1  \left \| \boldsymbol{\Omega}_k-\boldsymbol{\Omega}_k^{\prime} \right \|\\
        &\sigma \in \left \{ s,q \right \}. 
    \end{aligned}
    \end{equation*}
    b) Similarly, we assume that the second derivatives of the gradients of the loss function $\mathcal{L}_k \left (\cdot,\cdot\right )$ are \mbox{$\delta_2$-smooth}, i.e.,
    \begin{equation*}
    \begin{aligned}
        &\left \| \nabla_{\boldsymbol{\Omega}_k}^2 \mathcal{L}_k\left(f_{\boldsymbol{\Omega}_k^\prime}, \mathcal{D}_k^\sigma \right) - \nabla_{\boldsymbol{\Omega}_k}^2 \mathcal{L}_k\left(f_{\boldsymbol{\Omega}_k^{\prime}}, \mathcal{D}_k^\sigma \right) \right \| \leq \delta_2  \left \| \boldsymbol{\Omega}_k-\boldsymbol{\Omega}_k^{\prime} \right \|\\
        &\sigma \in \left \{ s,q \right \}. 
    \end{aligned}
    \end{equation*}
    
    \label{ass:smoothness}
\end{assumption}
\begin{assumption}[Lipschitz continuity]
    The  loss function $\mathcal{L}_k \left (\cdot,\cdot\right )$ For each client $k$ is  assumed to be $\gamma-$Lipschitz, i.e., for any \mbox{$\boldsymbol{\Omega}_k,\boldsymbol{\Omega}_k^{\prime} \in \boldsymbol\Xi$},
    \begin{equation*}
    \begin{aligned}
        &\left |  \mathcal{L}_k\left(f_{\boldsymbol{\Omega}_k^\prime}, \mathcal{D}_k^\sigma \right) -  \mathcal{L}_k\left(f_{\boldsymbol{\Omega}_k^{\prime}}, \mathcal{D}_k^\sigma \right) \right | \leq \gamma  \left \| \boldsymbol{\Omega}_k-\boldsymbol{\Omega}_k^{\prime} \right \|\\
        &\sigma \in \left \{ s,q \right \}. 
    \end{aligned}
    \end{equation*}
    \label{ass:lipsz}
\end{assumption}
\begin{assumption}
Given a localization task held by a client k, let $\ss _k^r$ be a data batch randomly sampled from the local dataset $\mathcal{D}_k$ at communication round $r$, we assume the following for the expected squared norm:
 \begin{equation*}
    \mathbb{E}_{\ss_k^r\sim\mathcal{D}_k}\left[ \left \| \boldsymbol\nabla_{\boldsymbol\theta}\mathcal{L}\left ( \boldsymbol\Omega_k, \ss_k^r\right ) \right \|^2\right] \leq \zeta^2, \; \forall k.
\end{equation*}
\label{ass:vbound}
\end{assumption}

\begin{lemma}[Inner-loop optimization]
    The $n^{th}$ gradient update during the inner loop optimization for a client k can be expressed, with respect to the initial gradients by:
    \begin{equation}
         \boldsymbol{\Omega}_k^{n} = \boldsymbol{\Omega}_k^{0} - \mu_k n\nabla_{\boldsymbol{\Omega}_k^{0}} \mathcal{L}_k\left(f_{\boldsymbol{\Omega}_k^{0}}, \mathcal{D}_k^s\right), \label{eq:nupdt}
    \end{equation}
    \begin{equation} 
        \text{that is: }\left\{\begin{matrix}
         \boldsymbol{\alpha}_k^{n} & = & \boldsymbol{\alpha}_k^{0} - \mu_k n\nabla_{\boldsymbol{\alpha}_k^{0}} \mathcal{L}_k\left(f_{\boldsymbol{\Omega}_k^{0}}, \mathcal{D}_k^s\right) \\
         \boldsymbol{\theta}_k^{n} & = & \boldsymbol{\theta}_k^{0} - \mu_k n\nabla_{\boldsymbol{\theta}_k^{0}} \mathcal{L}_k\left(f_{\boldsymbol{\Omega}_k^{0}}, \mathcal{D}_k^s\right) \\ 
         \boldsymbol{\beta}_k^{n} & = & \boldsymbol{\beta}_k^{0} - \mu_k n\nabla_{\boldsymbol{\beta}_k^{0}} \mathcal{L}_k\left(f_{\boldsymbol{\Omega}_k^{0}}, \mathcal{D}_k^s\right)
        \end{matrix}\right.
    \end{equation}
    \label{lem:inner}
\end{lemma}
\begin{proof}
    By \eqref{eq:innerloop}  we have the following:
    \begin{equation}
        \begin{matrix}
 \boldsymbol{\Omega}_k^{1} = & \boldsymbol{\Omega}_k^{0} - \mu_k \nabla_{\boldsymbol{\Omega}_k^{0}} \mathcal{L}_k\left(f_{\boldsymbol{\Omega}_k^{0}}, \mathcal{D}_k^s\right)\\ 

 \boldsymbol{\Omega}_k^{2} = & \boldsymbol{\Omega}_k^{1} - \mu_k \nabla_{\boldsymbol{\Omega}_k^{1}} \mathcal{L}_k\left(f_{\boldsymbol{\Omega}_k^{1}}, \mathcal{D}_k^s\right)\\ 

\boldsymbol{\Omega}_k^{2} = & \boldsymbol{\Omega}_k^{0} - \mu_k\left ( \nabla_{\boldsymbol{\Omega}_k^{0}} \mathcal{L}_k\left(f_{\boldsymbol{\Omega}_k^{0}}, \mathcal{D}_k^s\right)+ \nabla_{\boldsymbol{\Omega}_k^{1}} \mathcal{L}_k\left(f_{\boldsymbol{\Omega}_k^{1}}, \mathcal{D}_k^s\right)\right ) 
\end{matrix}
    \end{equation}
    Subsequently, the $n^{th}$ gradient update can be obtained via the accumulated gradient sum expressed as follows:
    \begin{equation}
         \boldsymbol{\Omega}_k^{n} =  \boldsymbol{\Omega}_k^{0} - \mu_k \sum_{i=1}^{n-1} \nabla_{\boldsymbol{\Omega}_k^{i}} \mathcal{L}_k\left(f_{\boldsymbol{\Omega}_k^{i}}, \mathcal{D}_k^s\right). \label{eq:gsum}
    \end{equation}
  Then, For $\mu_k \ll1,\;  k=1,2,\cdots, K$, in the infinite width regime, and following the neural tangent kernel (NTK) theory \cite{ntk}, the behavior of $\boldsymbol{\Omega}_k$ can be locally approximated as a linear function during training. Indeed, under these settings, The NTK analysis established that in the vicinity of $\boldsymbol{\Omega}_k^{0}$, $\mathcal{L}_k\left(f_{\boldsymbol{\Omega}_k}, \mathcal{D}_k^s\right)$ can be reasonably approximated by a linear function. This approximation is particularly accurate at the early stages of training when the weights are close to their initialization.
  Thus, the linearization of $\mathcal{L}_k\left(f_{\boldsymbol{\Omega}_k}, \mathcal{D}_k^s\right)$ around $\boldsymbol{\Omega}_k^{0}$ in \eqref{eq:gsum} holds the result in \eqref{eq:nupdt}.
\end{proof}

\begin{lemma}[Outer-loop optimization]
    The $r^{th}$ gradient update during the outer loop optimization for the federated server can be expressed, with respect to the initial gradients by:
    \begin{equation}
        \boldsymbol{\theta}^{r} = \boldsymbol{\theta}^0 - \eta r\mathbb{E}_k\bigg[   \nabla_{\boldsymbol{\theta}_k^{N,t} } \mathcal{L}_k(f_{\boldsymbol{\Omega}_k^{N,t}}, \mathcal{D}_k^q)\bigg].
        \label{eq:gsum2}
    \end{equation}
    \label{lem:outer}
\end{lemma}
\begin{proof}
    This lemma trivially extends Lemma \ref{lem:inner}.
\end{proof}
Let $\boldsymbol\theta^R = \boldsymbol\theta^m =$ be the pretrained meta-model return by Algorithm \ref{alg:femloc-train} after R communication rounds.
In the test environments, for any new localization tasks $\mathcal{T}_\kappa$, $\boldsymbol\theta^m$ is used as an initialization to train a personalized model $\boldsymbol\Omega_\kappa$ using Algorithm \ref{alg:femloc-test}.
Then, the performance of the FeMloc relies on how well $\boldsymbol\theta$ initialized the new task for accelerated convergence with few-shot learning. We thus consider two cases: (i) RI- standalone training of the new localization task with random initialization of the model weights $\boldsymbol\Omega^0 = \left[\boldsymbol\alpha^0,\boldsymbol\theta^0,\boldsymbol\beta^0\right]$, (ii) MI - training of the new localization task with initialization with the pretrained meta-model $\boldsymbol\Omega^m = \left[\boldsymbol\alpha^0,\boldsymbol\theta^m,\boldsymbol\beta^0\right]$.
\begin{proposition}
    Let $\kappa$ be a client assigned a new localization task in a test environment with initial model $\Omega^i$. The number of SGD steps required to achieve an $\varepsilon-$accurate solution $\boldsymbol\Omega_k^\varepsilon$ satisfies:
    \begin{equation*}
        (N_\kappa^i)^2 <  \frac{1}{\left ( \delta_1\mu   \right )^2}\left ( \frac{\varepsilon -2G^i}{\left \| \nabla_{\boldsymbol\Omega_k^i}\mathcal L_k^i \right \|^2}+1 \right )
    \end{equation*}
    with
    \begin{equation*}
         G^i = \left \|\nabla_{\boldsymbol\Omega_k^i}\mathcal L_k^i \right \|\cdot \left \|\nabla_{\boldsymbol\Omega_k^\varepsilon}\mathcal L_k^\varepsilon \right \|
    \end{equation*}
    \label{prop:steps}
\end{proposition}
\begin{proof}
    See Appendix \ref{apdx:prop2}.
\end{proof}
\begin{corollary}
For a client $k$ with a new localization task, given the two initialization schemes namely RI with $\boldsymbol\Omega_k^0$ and MI with $\boldsymbol\Omega_k^m$, we have the following:
    \begin{equation*}
        (N_\kappa^m)^2 - (N_\kappa^0)^2 < \frac{1}{\left ( \delta_1\mu \right )^2 }\left ( \frac{\varepsilon \Delta_2 -2G^o\left \|\nabla_{\boldsymbol\Omega_k^m}\mathcal L_k^m \right \|\Delta_1}{\Upsilon^2} \right )
    \end{equation*}
    where
    $$
    \begin{matrix}
 \Delta_1 &=& \left \|\nabla_{\boldsymbol\Omega_k^m}\mathcal L_k^m \right \|-\left \|\nabla_{\boldsymbol\Omega_k^0}\mathcal L_k^0 \right \|\\
\Delta_2 &=& \left \|\nabla_{\boldsymbol\Omega_k^m}\mathcal L_k^m \right \|^2-\left \|\nabla_{\boldsymbol\Omega_k^0}\mathcal L_k^0 \right \|^2\\
\Upsilon &=& \left \|\nabla_{\boldsymbol\Omega_k^m}\mathcal L_k^m \right \|\cdot \left \|\nabla_{\boldsymbol\Omega_k^0}\mathcal L_k^0 \right \|
\end{matrix}
    $$
    \label{cor:steps}
\end{corollary}
\begin{proof}
    See Appendix \ref{apdx:cor2}.
\end{proof}
\begin{remark}
    Thanks to the pretrained meta-model $\boldsymbol\theta^m$, the MI scheme yields an $\varepsilon_m-$accuracy.  
\end{remark}
\begin{proposition}
Let $N_\kappa^m$ be the required number of steps for the new client $\kappa$ to achieve an $\varepsilon_m-$accurate solution during the meta-testing phase in Algorithm \ref{alg:femloc-test} (MI scheme). Furthermore, we denote by $N_\kappa^0$ the required number of steps in the RI scheme, i.e., performing the training from scratch. It follows that:
\begin{enumerate}
    \item Ensuring a small enough $\varepsilon_m$ leads to faster convergence of the MI scheme, that is:
\begin{equation*}
    \varepsilon_m \leq \frac{\varepsilon^2}{\left ( 2\zeta \sqrt{\varepsilon}-\varepsilon \right )^2}\zeta^2 \rightarrow N_\kappa^m < N_\kappa^0
\end{equation*}
\item The difference in the required number of steps is upper-bounded as shown below:
\begin{equation*}
    (N_\kappa^m)^2 - (N_\kappa^0)^2 < \frac{\eta R\zeta }{\left ( \delta_1\Upsilon \mu \right )^2 }\left ( \delta_2\varepsilon +2\zeta \delta_1\sqrt{\varepsilon_m\varepsilon } \right )
\end{equation*}
\item The loss of the $\varepsilon-$accurate model satisfies
\begin{equation*}
    \mathcal L_k^\varepsilon < \frac{\gamma \mu }{2}\left (N_\kappa^0 \zeta +\sqrt{\varepsilon_m}N_\kappa^m \right )+\frac{\mathcal L_k^0+\mathcal L_k^m}{2}
\end{equation*}
\end{enumerate}
\label{prop:prop3}
\end{proposition}
\begin{proof}
    See Appendix \ref{apdx:prop3}.
\end{proof}
Consequently, the pretrained meta-model resulting from Algorithm \ref{alg:femloc-train} and used to initialize the new client $\kappa$ in  Algorithm \ref{alg:femloc-test} required fewer SGD iterations to achieve an $\varepsilon-$accurate localization model for the new environment. This theoretical result is backed up with experimental validation in the next section, for different indoor scenarios.

\section{Performance evaluation and discussion}
\label{sec:perfeval}

\subsection{Evaluation metrics}
To comprehensively assess FeMLoc's effectiveness, we evaluate its performance across various practical use cases using two key metrics: localization error (termed accuracy here) and adaptation speed to novel environments.
We measure localization accuracy using the mean distance error (MDE), which reflects the average distance between the predicted locations and the actual locations of IoT devices, and computed as follows:
 \begin{equation}
    MDE = \frac{1}{N}\sum_{i=1}^{N}\left \|  \mathbf{\hat y}_i-\mathbf{y}_i\right \|_2,
\end{equation}
with $\mathbf{y}_i$ and $\mathbf{\hat y}_i$ and the actual and predicted locations of target device $i$, respectively. $N$ is the total number of target devices to localize. This metric aligns well with regression tasks aiming for fine-grained localization, as opposed to zone classification where metrics like precision, recall, and F1-score would be more suitable.

Adaptation speed, on the other hand, gauges how rapidly the model adjusts to new environments by measuring its convergence rate towards a desired accuracy level.
We consider two approaches to calculate adaptation speed: (i) Accuracy-based adaptation speed $\Im(\mathcal A)$: This approach defines adaptation speed as the inverse of the number of training steps required to achieve a target accuracy value $\mathcal A$ in the new environment. It is formulated as:
\begin{equation}
    \Im(\mathcal A) = \mathbb{E}\left [ \max_n \left [ \frac{1}{b\times \mathcal N(\mathcal A(n))} \right ] \right ].
\end{equation}
Here, $\mathbb{E}$ denotes the expectation function, $n$ represents the number of gradient steps, $b$ is the batch size, and $ \mathcal N(\mathcal A(n))$ represents the number of training steps needed to reach accuracy $\mathcal A$. With $\mathcal A(n) = MDE\mid n$, $\mathcal N(\mathcal A(n)) = \mathcal A^{-1}(n) = n \mid MDE$.
This approach focuses on the time it takes for the model to converge to a specific accuracy threshold. A higher value of $\Im(\mathcal A)$ indicates superior model performance.
(ii) Step-based adaptation speed ($\Im(n^*)$): Alternatively, we can assess the model's accuracy after a fixed number of training steps $n^*$. In this case, adaptation speed is defined as the inverse of the $MDE$ achieved at step $n^*$. This is mathematically expressed as:
\begin{equation}
     \Im(n^*) =   \frac{1}{b}  \times \mathcal \mathcal A(n^*).
\end{equation}
This approach provides a snapshot of the model's accuracy at a specific point during the adaptation process. In this context, a lower value of $\Im(\mathcal n^*)$ signifies superior model performance.

By employing both localization accuracy and adaptation speed, we can gain a comprehensive understanding of FeMLoc's ability to provide accurate and efficient indoor localization under diverse real-world conditions.

 \begin{figure*}[t]
\centering
  \begin{subfigure}[t]{.32\textwidth}
    \centering
    \includegraphics[width=\linewidth, height=0.2\textheight]{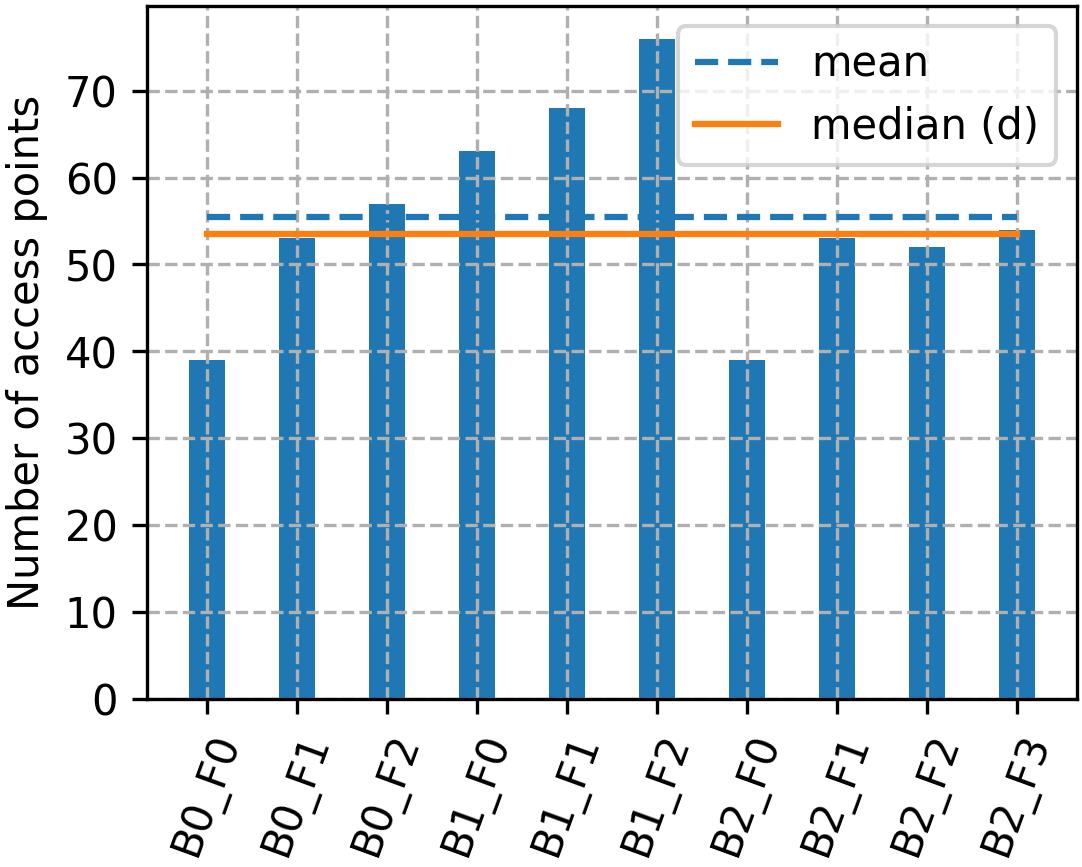}
    \caption{}
    \label{fig:distribution}
  \end{subfigure}
  \hfill
  \begin{subfigure}[t]{.34\textwidth}
    \centering
    \includegraphics[width=\linewidth, height=0.2\textheight]{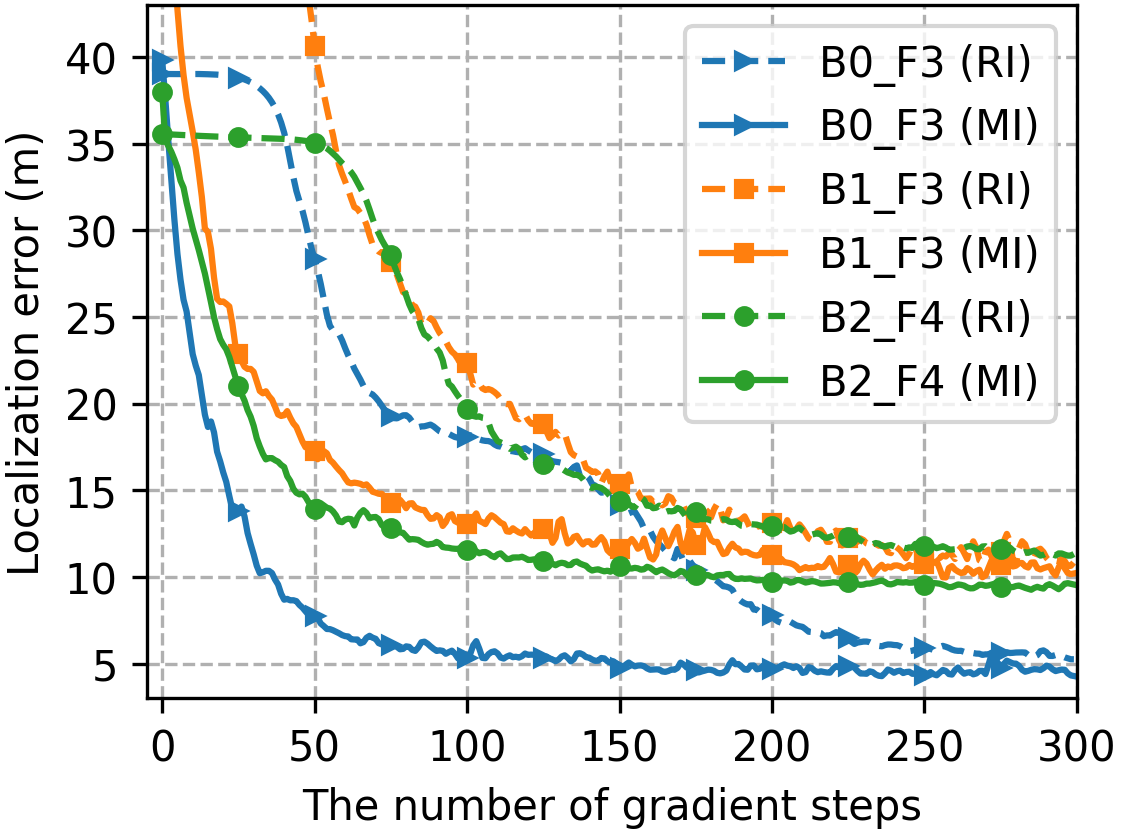}
    \caption{}
    \label{fig:curves}
  \end{subfigure}
  \hfill
  \begin{subfigure}[t]{.32\textwidth}
    \centering
    \includegraphics[width=\linewidth, height=0.2\textheight]{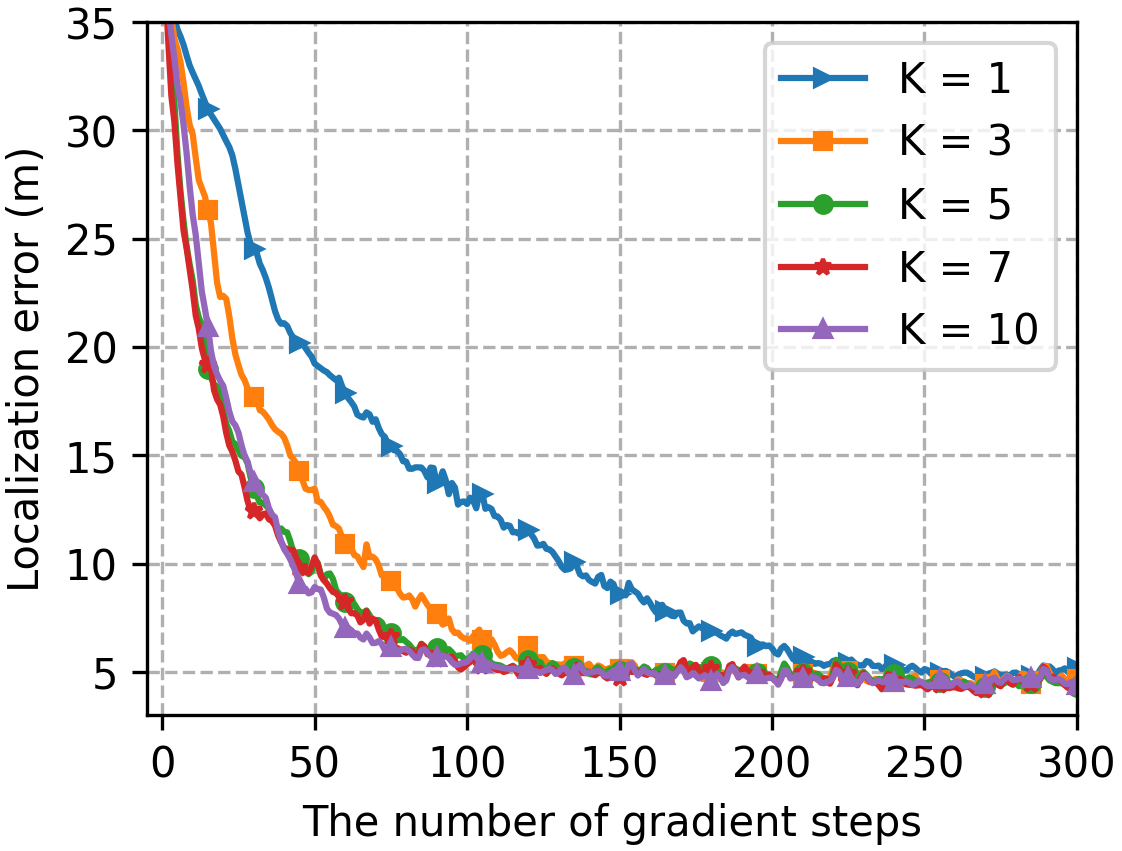}
    \caption{}
    \label{fig:task}
  \end{subfigure}
  \caption{Performance analysis of FeMLoc on UJIIndoorLoc dataset \cite{ujiindoorloc}: (a) Input signal space distribution of localization training tasks \{$m_k$\} (b) Learning curves of test tasks (c) Learning curves with varying number K of training task.}
  \label{fig:perf}
\end{figure*}

\begin{figure}[!t]
    \centering
    \includegraphics[scale=0.8]{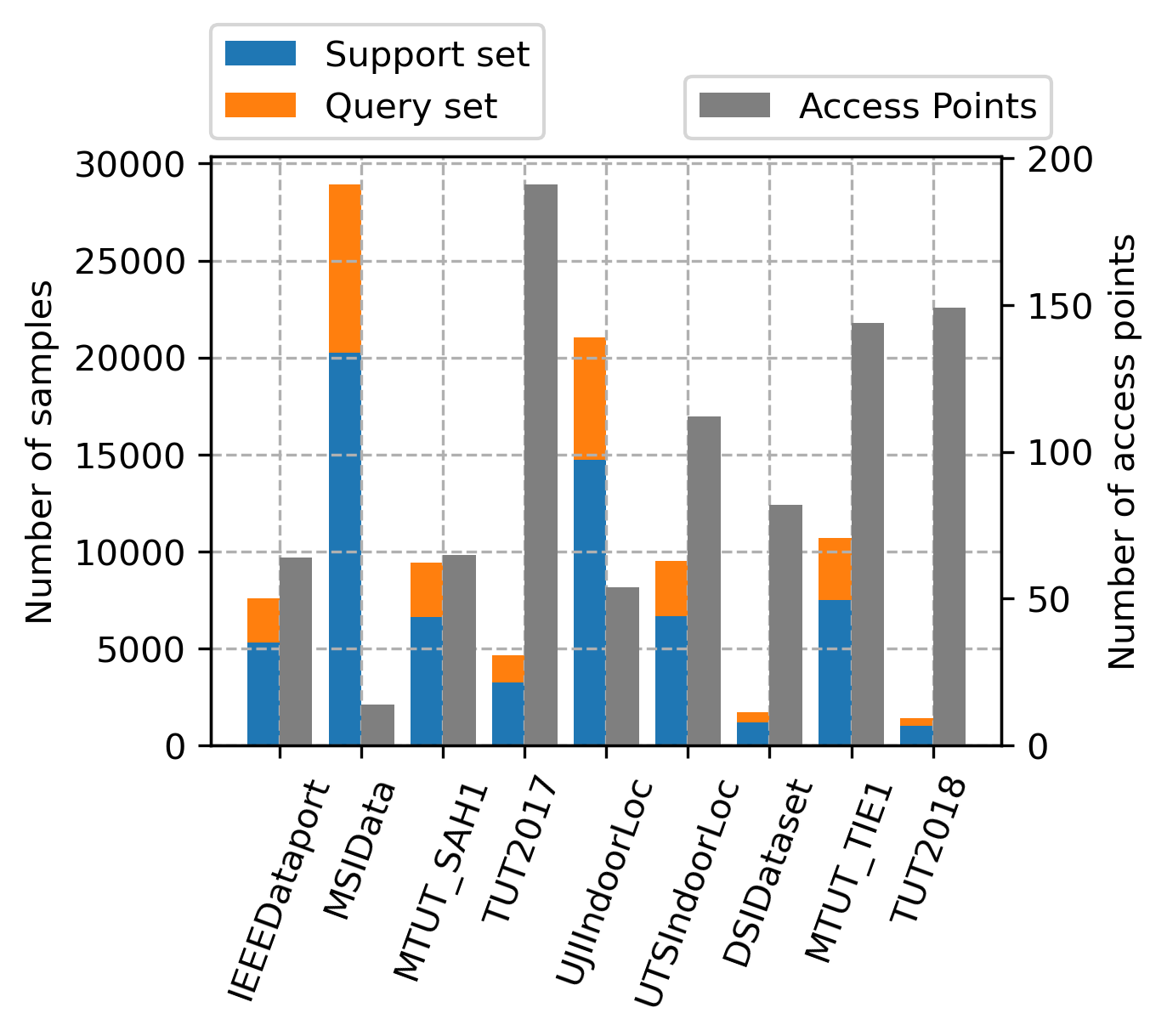}
    \caption{Data distribution}
    \label{fig:gendist}
\end{figure}
\subsection{Multi-floor Learning: Evaluation on UJIIndoorLoc dataset (EXP1)}
\label{subsec:scalability}
In this part, we present the performance analysis of our proposed framework using a publicly available RSSI fingerprinting dataset obtained from a real-world experiment. Specifically, we utilize the UJIIndoorLoc dataset \cite{ujiindoorloc}, which was collected on a campus encompassing three buildings, each comprising 4 to 5 floors. For the purpose of our experiment, each floor represents an individual indoor localization task.
The individual floor data normalization has been conducted following the  preprocessing method presented earlier in Section \ref{sec:preproc}.

In this study, the experimental setup involved dividing the floors of the buildings into meta-testing and meta-training sets. Three floors, one in each building, were reserved for the meta-testing phase, while the remaining ten floors were used for the meta-training. To represent the localization tasks each assigned to a client for federated training, a naming convention \textit{B\{i\}\_F\{j\}} was adopted, where $j$ denoted the floor ID and $i$ representing the ID of the building containing the floor $j$. Fig. \ref{fig:distribution} visually demonstrates the heterogeneity of indoor localization tasks in terms of the number of APs, highlighting the need for model personalization in the proposed collaborative learning framework.

The configurations of the general models are shown in TABLE \ref{tab:model}, where the input space of each client is set by the number of APs appearing in its localization task. A batch size of 32 was selected for all local gradient updates, and the outer learning rate ($\eta$) was set to 0.001.
Additionally, the update frequency, defined as the number of local gradient steps per communication round, was set to $N_k = N =5$ for all clients, with the mean squared error (MSE) as the loss function. Note that these hyperparameter configurations, including those listed in TABLE \ref{tab:model}, were used for all the practical use cases evaluated in our framework.

Note that to select these hyperparameters for FeMLoc, we adopted a two-stage approach. First, we drew inspiration from existing research on similar RSSI datasets (TABLE \ref{tab:datasets}) to establish a solid foundation for our hyperparameters' selection. Second, we performed further fine-tuning during a preliminary phase to obtain the final chosen hyperparameters that we reported for transparency and reproducibility.

\begin{remark}
\label{rem:split}
The experiment was conducted through the repetitive random splitting of tasks into training and testing sets (in a ratio of 10:3). Despite these multiple iterations, consistent outcomes were consistently obtained. For clarity in this presentation, only one experiment result has been reported.
\end{remark}
Following the training phase, comprising a total of R=1000 communication rounds, the global meta-model is utilized to initialize the model of the test tasks, for which the training performance is presented in Fig. \ref{fig:perf}. 
It illustrates the learning curves for various test tasks, highlighting the impact of meta-initialization compared to random initialization. Each test task has two curves: one for a randomly initialized (RI) model and another for a model initialized with the pre-trained global meta-model (MI). The results clearly demonstrate that meta-initialization (MI) significantly accelerates convergence compared to random initialization (RI) across all test tasks, thereby achieving the principal objective of the proposed framework. Furthermore, meta-initialization is observed to produce enhanced accuracy, particularly in scenarios where the test task is characterized by limited data samples. 

Detailed quantitative results are provided in the EXP1 section of TABLE~\ref{tab:adaptationspeed}. For example, for a target accuracy of $5m$, FeMLoc (MI) achieves adaptation $\Im (\mathcal{A})$ 67.74\% faster than RI in the B0\_F3 environment with 71.02\% improvement in the localization accuracy.
These evaluations demonstrate FeMLoc's ability to rapidly adapt to new indoor environments, facilitating large-scale deployment of indoor localization solutions with minimal retraining efforts.

Moreover, an exploration of the convergence speed using meta initialization concerning the number of training tasks was undertaken, and the corresponding outcome is displayed in Fig. \ref{fig:task}. As anticipated, an increase in the number of clients engaged in the meta-training process leads to expedited adaptation during the meta-testing phase. It is worth noting that K = 1 is equivalent to transfer learning where knowledge acquired through a given floor of the indoor environment is transferred to a new floor in the same indoor environment.
We further investigated the impact of the number of training tasks (clients) on convergence speed using meta-initialization. Results are presented in Fig.~\ref{fig:task}. As expected, increasing the number of clients participating in the meta-training process leads to faster adaptation during the meta-testing phase. 

\begin{table}[!t]
\caption{Client model configuration.}
\begin{tabular}{ll|llll}
 &               & encoder & decoder & local meta-model & mapper  \\
 \hline
 & input layer   & $m_k$    & d=50    & d=50       & n=32    \\
 & hidden layer  & 1024    & 1024    & 256x128x64 & 64x32   \\
 & output        & d=50    & $m_k$     & n=32       & p=2     \\
 & optimizer     & ADAM    & ADAM    & ADAM       & ADAM    \\
 & learning rate $\mu_k$ & 0.0095  & 0.0095  & 0.0005   & 0.0005
\end{tabular}
\label{tab:model}
\end{table}

As a result, FeMLoc effectively tackles scalability challenges for indoor localization across large, multi-building campuses. It demonstrates strong adaptability and robust performance in floor-level localization tasks in diverse architectural environments. It is worth noting that this 10-limit in the scalability study was constrained by the availability of localization datasets (number of clients $K=1,2,\cdots,10$).
\begin{figure*}[!t]
\centering
  \begin{subfigure}[t]{.32\textwidth}
    \centering
    \includegraphics[width=\linewidth, height=0.2\textheight]{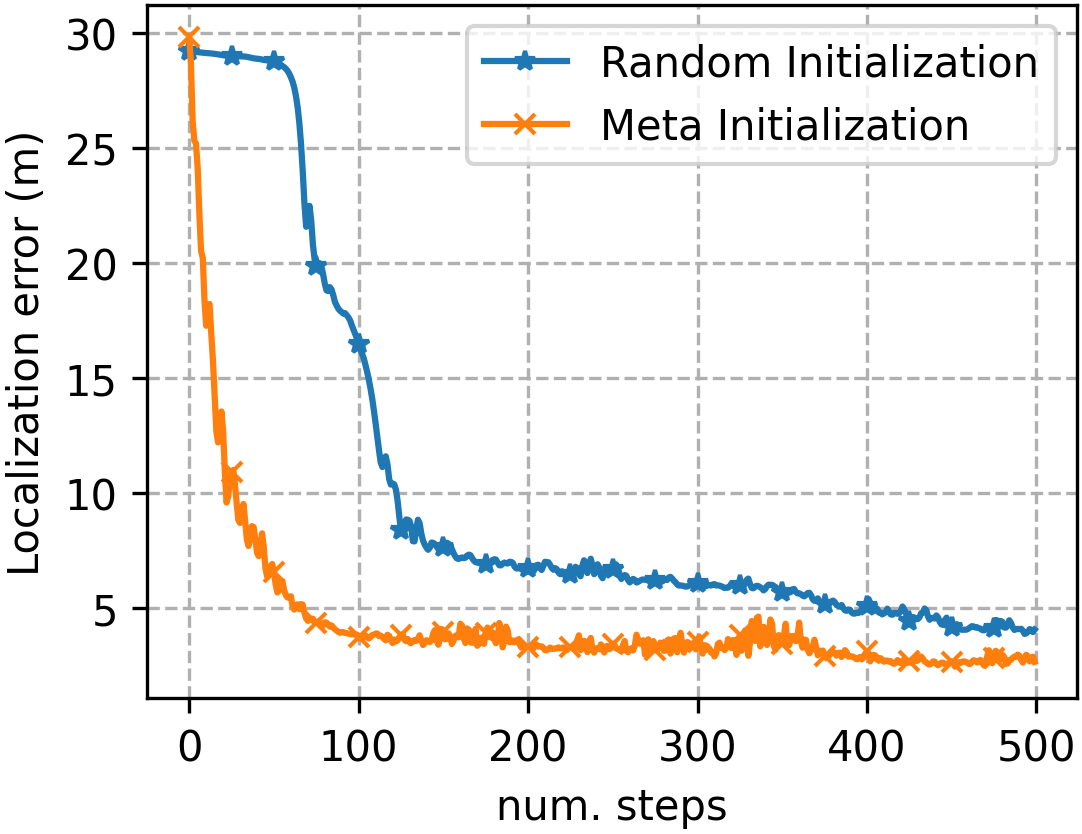}
    \caption{}
    \label{fig:test1}
  \end{subfigure}
  \hfill
  \begin{subfigure}[t]{.34\textwidth}
    \centering
    \includegraphics[width=\linewidth, height=0.2\textheight]{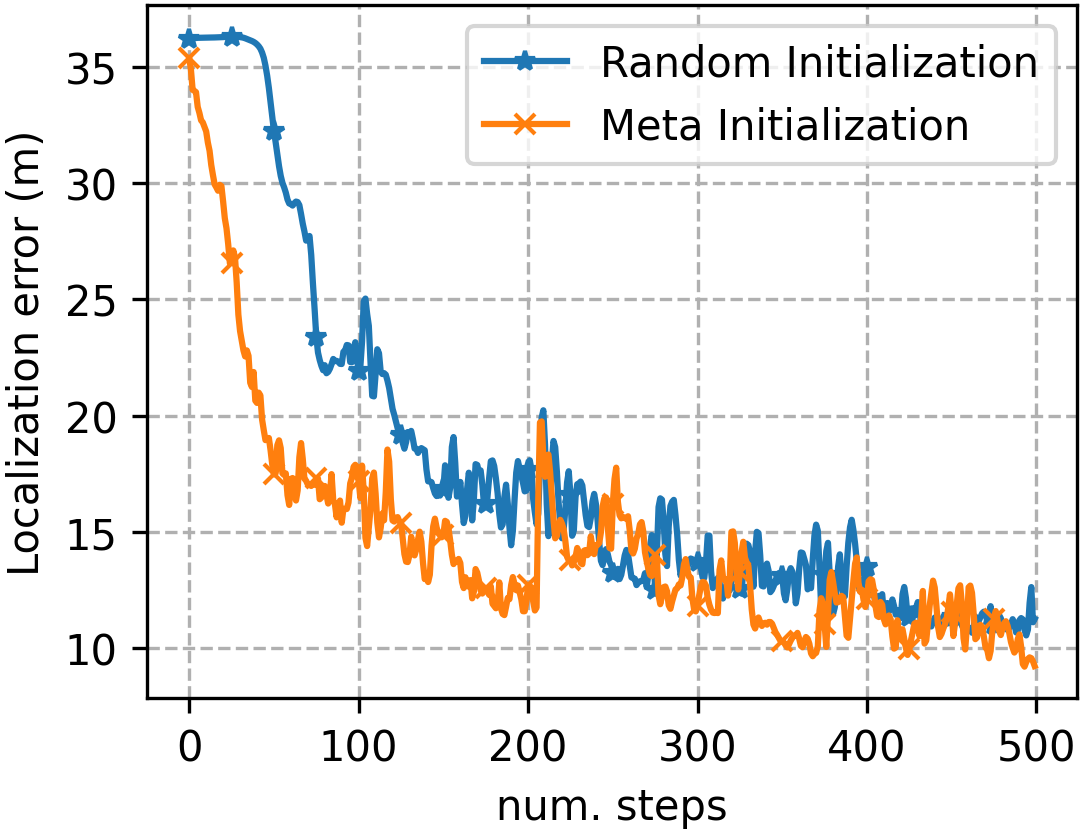}
    \caption{}
    \label{fig:test2}
  \end{subfigure}
  \hfill
  \begin{subfigure}[t]{.32\textwidth}
    \centering
    \includegraphics[width=\linewidth, height=0.2\textheight]{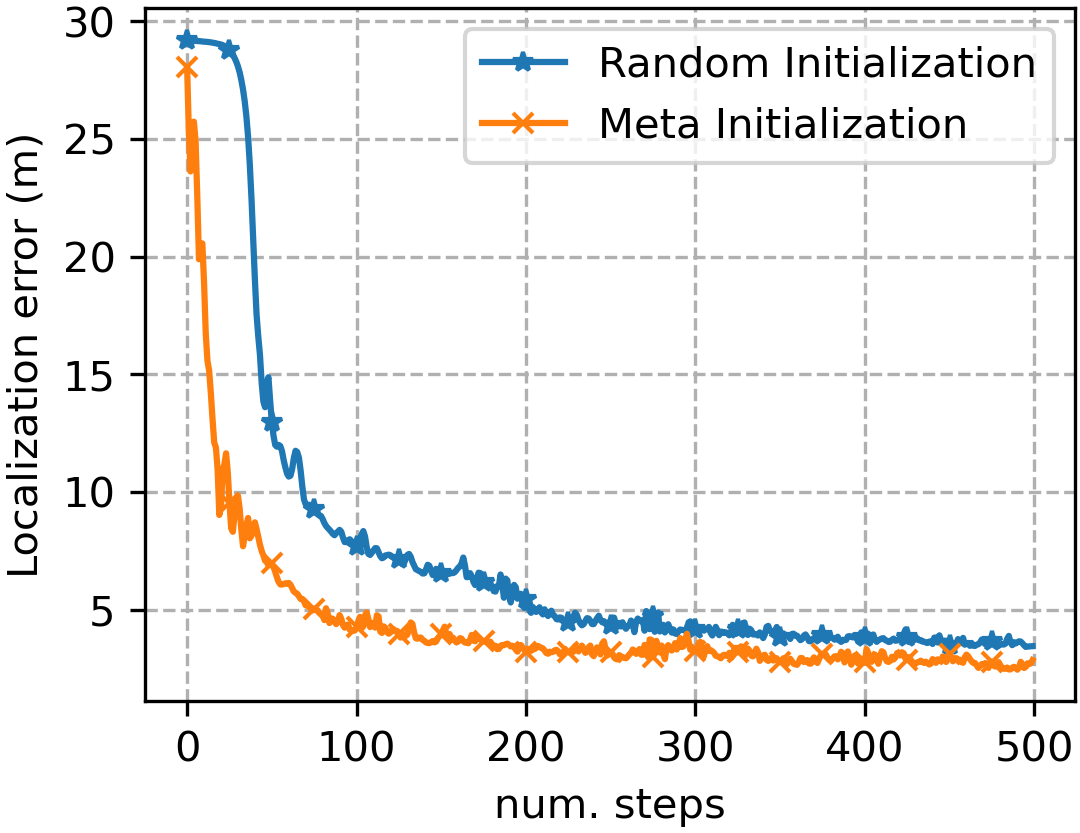}
    \caption{}
    \label{fig:test3}
  \end{subfigure}
  \caption{Performance analysis of FeMLoc on multi-environment datasets: (a) DSIDataset (b) TUT2018 (c) MTU\_TIE1.}
  \label{fig:generalization}
\end{figure*}

\begin{table*}[!t]
\centering
\caption{Performance evaluation of FeMLoc in different experimental setups.} 
\begin{tabular}{|c|c|c|c|c|c|c|cc|c|c|c|c|c|}
\hline

Experiment & Environment & $\mathcal A (m)$ & Method    & n & $\Im(\mathcal A) \times 10^{-3}$   &  \%$\uparrow$ && &$n^*$& Method& $\mathcal A (m)$ &$\Im(n^*)$&  \%$\uparrow$\\ \hline \hline
                      
\multirow{12}{*}{EXP1} & \multirow{6}{*}{B0\_F3} & \multirow{2}{*}{5}  & MI     & 100  & 0.31   &  67.74   && &\multirow{2}{*}{50} & MI&7.5 &0.23 &72.63\\ \cline{4-7} 
                      &                         &                     & RI     & 310  & 0.10   &   --  && & &RI &27.4 &0.86 &--\\ \cline{3-7} \cline{10-14} 
                      &                         & \multirow{2}{*}{10} & MI      & 30  & 1.04  & 82.86   && &\multirow{2}{*}{100} &MI &5.1 &0.16 &71.02\\ \cline{4-7} 
                      &                         &                     & RI      & 175  & 0.18  & --   && & &RI &17.6 &0.55 &--\\ \cline{3-7} \cline{10-14} 
                      &                         & \multirow{2}{*}{15} & MI      & 25  & 1.25   & 83.11    && &\multirow{2}{*}{150} &MI &4.9 &0.15 &66.89\\ \cline{4-7} 
                      &                         &                     & RI      &175  & 0.21   &  --  && & &RI &14.8 &0.46 &--\\ \cline{2-7} \cline{10-14}
                      & \multirow{2}{*}{B1\_F3} & \multirow{2}{*}{12}  & MI      & 150  &0.21    &34.78   && &\multirow{2}{*}{100} &MI &13.1 &0.41 &41.26\\ \cline{4-7} 
                      &                         &                     & RI      &230   &0.14    & -- && & &RI &22.3 &0.7 &--\\ \cline{3-7}  \cline{10-14} 
    \cline{2-7} 
                      & \multirow{2}{*}{B2\_F4} & \multirow{2}{*}{10}  &MI       &175   & 0.18   &  50.56    && &\multirow{2}{*}{100} &MI &12.2 &0.38  &38.38\\ \cline{4-7} 
                      &                         &                     & RI      &354   & 0.09  & --    && & &RI &19.8 &0.62 &--\\ \cline{3-7} \cline{10-14}  
  \hline \hline
\multirow{6}{*}{EXP2} & \multirow{2}{*}{DSIDataset} & \multirow{2}{*}{5}& MI     &60   & 0.52   &  82.21  && &\multirow{2}{*}{100} &MI &3.2 &0.1 &80.95\\ \cline{4-7} 
                      &                         &                     & RI     &380   &0.08     &--    && & &RI &16.8 &0.52 &--\\ \cline{3-7}  \cline{10-14} 
      \cline{2-7} 
                      & \multirow{2}{*}{TUT2018} & \multirow{2}{*}{12}  & MI     &200   &0.16    &48.05    && &\multirow{2}{*}{100} &MI &15.3 &0.48  &30.77 \\ \cline{4-7} 
                      &                         &                     & RI      &385   &0.08   &  --  && & &RI &22.1 &0.69  &-- \\ \cline{3-7} \cline{10-14} 
           \cline{2-7} 
                      & \multirow{2}{*}{MTU\_TIE1} & \multirow{2}{*}{5} &MI      & 72  & 0.43   & 64.0   && &\multirow{2}{*}{100} &MI &4.2 &0.15  &39.23\\ \cline{4-7} 
                      &                         &                     &  RI      & 200  & 0.16  & --    && & &RI &7.9 &0.28  &--\\ \cline{3-7} \cline{10-14}  
          \hline \hline
\multirow{6}{*}{EXP3} & \multirow{2}{*}{13th month} & \multirow{2}{*}{2.5}  & MI     &100   &0.31    &66.67   && &\multirow{2}{*}{200} &MI &2.24 & 0.07  &25.83\\ \cline{4-7} 
                      &                         &                     & RI     &300   & 0.1   & --   && & &RI &3.02 & 0.09 &--\\ \cline{3-7}  \cline{10-14} 
        \cline{2-7} 
                      & \multirow{2}{*}{14th month} & \multirow{2}{*}{2.5}  &MI        &85   &0.37    &73.85    && &\multirow{2}{*}{200} &MI &2.22 &0.07  &33.73\\ \cline{4-7} 
                      &                         &                     &RI        & 325  &0.1    & --   && & &RI &3.35 &0.1  &--\\ \cline{3-7}  \cline{10-14} 
        \cline{2-7} 
                      & \multirow{2}{*}{15th month} & \multirow{2}{*}{2.5}  & MI       &85   &0.37    & 78.75   && &\multirow{2}{*}{200} &MI &2.21 & 0.07& 35.94\\ \cline{4-7} 
                      &                         &                     &RI        & 400  &0.08    & --   && & &RI &3.45 &0.11  &--\\ \cline{3-7}  \cline{10-14} 
           \hline
\end{tabular}
\label{tab:adaptationspeed}
\end{table*}

\subsection{Multi-environment Learning (EXP2)}
\label{subsec:generalization}
In this section, we present an evaluation of our framework's performance using experimental datasets from diverse campuses and distinct indoor spaces.
This evaluation aims to demonstrate the framework's generalization capabilities and its ability to adapt to novel settings.

For this experiment, we consider training datasets comprising samples collected from various indoor environments, each representing unique architectural characteristics and environmental conditions. Fig.~\ref{fig:gendist} illustrates the distribution of these datasets, highlighting the significant heterogeneity across the different indoor environments. This diversity motivates the general architecture proposed in FeMLoc, which is designed to be adapted to such variations. We run the experiment using the same model architecture presented in TABLE \ref{tab:model} with the hyper-parameters of Section \ref{subsec:scalability}.
Fig.~\ref{fig:generalization} showcases the effectiveness of our approach by comparing FeMLoc (MI) with a baseline model initialized randomly (RI). Note that Remark~\ref{rem:split} regarding the data split ratio still applies here.

The results in Fig.~\ref{fig:generalization} clearly demonstrates FeMLoc's (MI) superior ability to generalize and rapidly adapt to new indoor environments. This translates to successful localization even when encountering unseen environments. Detailed quantitative results are provided in the EXP1 section of TABLE~\ref{tab:adaptationspeed}. Notably, FeMLoc achieves up to 82.21\% faster adaptation compared to the baseline while guaranteeing up to 80.95\% improvement in localization accuracy after only 100 steps. These findings underscore the framework's potential for real-world applications and highlight its contribution towards advancing indoor positioning solutions that transcend the boundaries of individual indoor environments.

Furthermore, it worth recalling that our proposed  framework represents an advanced form of transfer learning, amplifying its significance in addressing indoor localization challenges across varying environments. Traditional transfer learning applies a pre-trained model from one environment (source domain) to a new environment (target domain). However, FeMLoc leverages knowledge from various environments (domains) while accounting for their inherent differences. This enables more general knowledge transfer across diverse indoor spaces, resulting in faster adaptation. By dynamically updating the meta-model through localized training, our framework goes beyond conventional transfer learning, showcasing its ability to refine and fine-tune knowledge for each specific environment, thereby achieving superior generalization and adaptability. This enhanced transfer learning paradigm enriches the framework's potential to excel in complex and dynamic indoor positioning scenarios, underscoring its role as a pioneering solution in the realm of indoor localization.



\subsection{Dynamic environment Learning (EXP3)}
\label{subsec:dynamic}
In this section, we present a comprehensive analysis of our framework's performance in an experiment utilizing multiple datasets collected within the same indoor environment over time. The objective of this experiment is to showcase the framework's capability to minimize the data collection efforts necessary for recalibrating the localization model in a dynamic indoor environment, thereby enhancing its long-term adaptability.

We outline the experimental configuration, where multiple datasets spanning different time periods within the same indoor environment are employed for training. Each dataset encapsulates variations in environmental conditions, APs characteristics, and other factors that naturally evolve over time, thus capturing the evolving dynamics of the environment. 

The meta-model is trained on this historical data, and its performance is subsequently evaluated on newly collected datasets within the same indoor environment.
We used the experimental dataset presented in \cite{ltmsupport} collected over a period of 15 months at one-month intervals.
We consider the data up to the $12^{th}$ month for the meta-training while the last 3 months' data were kept for the meta-testing.
We run the experiment using the same model architecture presented in TABLE \ref{tab:model} with the hyper-parameters of Section \ref{subsec:scalability}, and results are shown in Fig. \ref{fig:dynamic}
where we compare the localization performance of our meta-model (MI) on newly collected datasets with that of conventional method (RI) requiring full data recalibration.

Our solution achieves demonstrably faster convergence and higher accuracy. This translates to a significant reduction in the need for continuously collecting and annotating large amounts of data over time. For example, to achieve an accuracy of $2.5m$, FeMLoc (MI) only requires at most 100 SGD steps, whereas the baseline method (RI) needs 300-400 steps. This translates to more than 66.67\% faster adaptation speed, while guaranteeing up to 35.94\% improvement in localization accuracy after 200 steps. 
The EXP3 section of TABLE~\ref{tab:adaptationspeed} provides further quantitative analysis supporting this result.

While adaptation speed was the primary focus of this framework, we also evaluated FeMLoc's localization accuracy to provide a more comprehensive assessment. We benchmarked FeMLoc against state-of-the-art methods: support vector machine regression (SVR), k-nearest neighbors (KNN) regression ($k = 11$), a conventional deep neural network (RI), and transfer learning (TL). It's important to note that TL, in this context, represents a special case of FeMLoc (MI) applied with only one FL device.

We used one of the test datasets ($15^{th}$ month dataset) for this benchmark analysis. The results are shown in Fig.~\ref{fig:bchmk} representing the cumulative distribution functions (CDF) of the localization accuracy of the different models. As it can be seen, FeMLoc (MI) achieves superior performance compared to the other algorithms, with a localization accuracy or mean distance error (MDE) of 1.798 meters. This is notably better than the 1.866 meters achieved by TL, the next best performing method. This evaluation demonstrates that FeMLoc not only provides fast adaptation in dynamic indoor environments but also slightly improves localization accuracy.

Overall, this performance analysis substantiates that FeMLoc excels as an innovative solution to reduce the data collection burden for recalibrating indoor localization models over time. By leveraging historical datasets and exhibiting superior adaptability, the framework presents a transformative approach that holds promise for long-term, sustainable indoor positioning applications.

\begin{figure}[!t]
    \centering
    \includegraphics[scale=0.8]{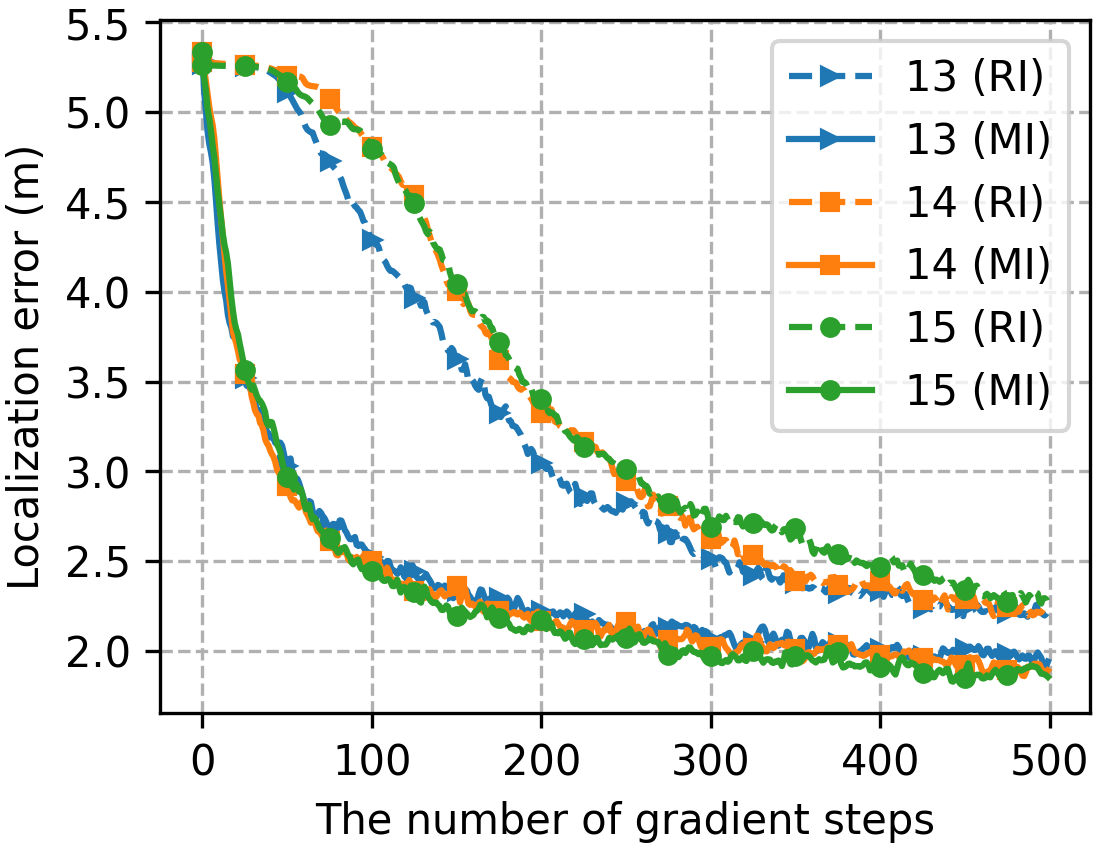}
    \caption{Learning curves in long-term support model adaption}
    \label{fig:dynamic}
\end{figure}

\begin{figure}[!t]
    \centering
    \includegraphics[scale=0.8]{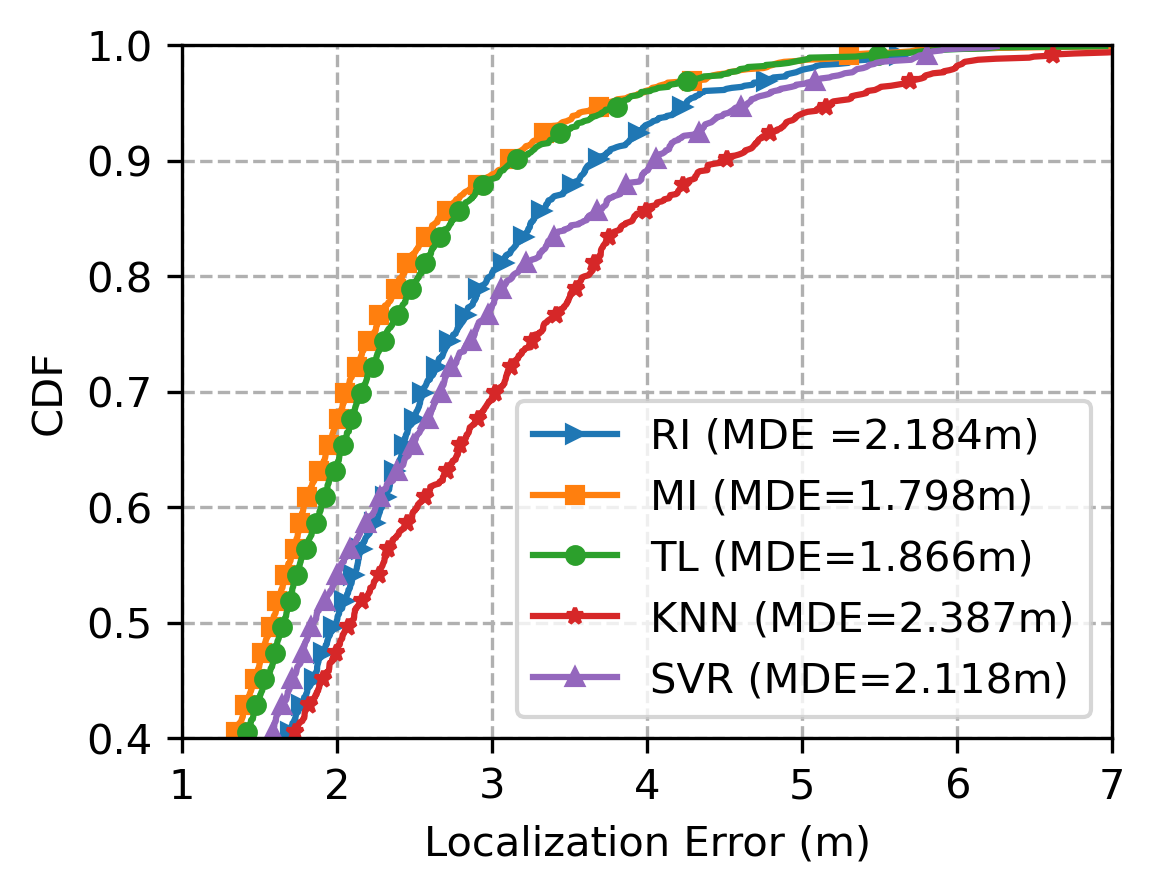}
    \caption{Performance comparison of FeMLoc in terms of localization accuracy in dynamic indoor environments}
    \label{fig:bchmk}
\end{figure}

\section{Limitations and future work}
\label{sec:limit}
In this section, we discuss the limitations of our proposed FeMLoc framework and outline potential avenues for future research to address these limitations and further advance the field of federated meta-learning for IoT applications.
\subsection{Privacy considerations}
While FeMLoc leverages privacy-preserving mechanisms inherent in FL to safeguard sensitive data during model training and aggregation, there remain challenges in ensuring comprehensive privacy protection, particularly in scenarios involving highly sensitive data or stringent privacy regulations. While FL offers advantages in privacy preservation by keeping training data on local devices, the exchange of model updates during meta-learning introduces new privacy risks. The meta-model, aggregating knowledge extracted from local models, may inadvertently reveal information about the specific training data used on individual devices. 

Furthermore, ensuring meta-model consistency across all devices poses another challenge. Malicious FL participants can intentionally manipulate the model, deviating it from learning the genuine data patterns from various devices, leading to potential vulnerabilities.
To address these challenges, future research could explore advanced privacy-enhancing techniques, such as homomorphic encryption and secure multi-party computation \cite{flprivacy}, to further fortify privacy guarantees in FeMLoc.

\subsection{Generalization to other IoT applications}
The core principles of FeMLoc, particularly its federated meta-learning approach, hold significant promise for extending its applicability to diverse IoT applications beyond indoor localization. However, generalizing FeMLoc requires careful consideration of the specific data characteristics, application requirements, and environmental conditions inherent to different use cases.
Future research endeavors could focus on investigating the adaptation of FeMLoc to various IoT scenarios, such as predictive maintenance, health monitoring, and environmental sensing. This exploration entails the examination of domain-specific model architectures, tailored data preprocessing techniques, and relevant performance evaluation metrics to ensure FeMLoc's effectiveness across different application domains.

Moreover, addressing the challenge of scalability and efficiency for large-scale deployments is crucial for the widespread adoption of FeMLoc. Research efforts in this area could explore techniques such as device selection strategies, gradient compression methods, and the design of model architectures with reduced computational requirements. These endeavors aim to enhance FeMLoc's scalability and efficiency when deployed across a significantly larger number and wider variety of resource-constrained IoT devices.

 \subsection{ Communication Aspects}
The communication overhead associated with exchanging model updates during federated learning poses a significant challenge, particularly when scaling to large numbers of devices. Moreover, the reliability and latency of communication channels can significantly impact the performance and stability of the learning process. To address this challenge, one potential direction involves exploring the integration of edge computing and edge intelligence techniques \cite{fledge}. By offloading computation and communication tasks from resource-constrained IoT devices to edge servers, the efficiency and scalability of FeMLoc can be enhanced.
Another interesting direction is the extension of FeMLoc over heterogeneous networks, wherein the feasibility and challenges of deploying FeMLoc in scenarios with diverse network conditions, such as cellular networks, Wi-Fi, and Bluetooth, are explored. This entails the development of robust communication protocols capable of effectively operating across heterogeneous network types and adapting to varying levels of connectivity.

Additionally, the development of communication-adaptive learning algorithms represents an intriguing area of investigation. These algorithms dynamically adapt to the characteristics of the communication channel, such as bandwidth and latency \cite{flcom}. This adaptation may involve adjusting the frequency of update exchanges or implementing strategies to mitigate potential packet loss or delays, thereby optimizing communication efficiency.
Furthermore, research efforts can focus on exploring efficient communication techniques to reduce communication overhead, including gradient compression, knowledge distillation, and federated pruning. These techniques aim to minimize the size of model updates transmitted between devices and the central server, thereby improving scalability and resource efficiency in  FeMLoc.

\section{Conclusion}
\label{sec:conclusion}


In this paper, we introduced FeMLoc, a novel federated meta-learning framework tailored for indoor localization in IoT environments. Leveraging the synergy between federated learning and meta-learning, FeMLoc addresses the challenges of data heterogeneity, privacy preservation, and model adaptability inherent in indoor localization tasks.
Through extensive experimentations and analysis, we demonstrated the effectiveness and scalability of FeMLoc in various indoor environments. Our results highlight FeMLoc's ability to achieve superior localization accuracy with minimal data collection and model adaptation efforts. Moreover, the collaborative nature of federated learning ensures privacy preservation while enabling knowledge sharing across diverse IoT devices.
While FeMLoc represents a significant advancement in federated meta-learning for IoT applications, there exist several challenges and opportunities for future research. By addressing the limitations discussed and exploring the proposed future directions, we can further enhance the scalability, communication efficiency, and privacy-preserving capabilities of FeMLoc, ultimately advancing its applicability and impact in the field of IoT-enabled systems and services.

\appendices


\section{}
\subsection{Proof of Proposition \ref{prop:steps}}
\label{apdx:prop2}
\begin{lemma}
Let $U$ and $\epsilon$ be positive and strictly positive real numbers, respectively. If $U < \epsilon$, then:
\begin{enumerate}
    \item there exists a positive real number $\epsilon_0$ such that \mbox{$U \leq \epsilon_0 < \epsilon$}.
    \item there exists a positive real number $\epsilon_1$ such that \mbox{$U < \epsilon_1 \leq \epsilon$}.
\end{enumerate}
    \label{lem:xbound}
\end{lemma}
\begin{proof}
 Given $U < \epsilon$, \\
 \begin{enumerate}
     \item we define $\epsilon_0$ as the midpoint between $U$ and $\epsilon$, that is  $\epsilon_0 = \frac{U + \epsilon}{2}$.
It follows that: (i) $\epsilon_0 > U$ because it is the midpoint between $U$ and $\epsilon$.
(ii) $\epsilon_0 < \epsilon$ for the same reason.
\item consider the positive real number $\epsilon_1$ defined as
 \mbox{$\epsilon_1 = U + \frac{\epsilon - U}{2}$}. Then: (i) $ U < \epsilon_1 $ because we're adding a positive value to $U$.
    (ii) $ \epsilon_1 \leq \epsilon $ because we're taking the midpoint between $U$ and $\epsilon$.
 \end{enumerate}
Therefore, the conditions of the lemma are satisfied.
\end{proof}
To simplify the notations, let $\mathcal{L}_\kappa^n = \mathcal{L}_k(f_{\boldsymbol{\Omega}_\kappa^{n}}, \mathcal{D}_\kappa^q)$ and $\nabla_{\boldsymbol{\Omega}_\kappa^{n} } \mathcal{L}(f_{\boldsymbol{\Omega}_\kappa^{n}}, \mathcal{D}_k^q) = \nabla\mathcal{L}_\kappa^n$.

By the reverse triangle inequality
\begin{equation*}
    \left |  \left \| \nabla\mathcal{L}_\kappa^\varepsilon  \right \| - \left \| \nabla\mathcal{L}_\kappa^i \right \|\right | \leq \left \| \nabla\mathcal{L}_\kappa^\varepsilon- \nabla\mathcal{L}_\kappa^i \right \|
\end{equation*}
By squaring both sides and using Assumption \ref{ass:smoothness}.a, we have
\begin{equation}
    \left \| \nabla\mathcal{L}_\kappa^\varepsilon \right \|^2\leq \delta_1^2\left \| \boldsymbol\Omega_\kappa^\varepsilon- \boldsymbol\Omega_\kappa^i \right \|^2+2\left \| \nabla\mathcal{L}_\kappa^\varepsilon \right \|\cdot \left \| \nabla\mathcal{L}_\kappa^i \right \|-\left \| \nabla\mathcal{L}_\kappa^i \right \|^2.
    \label{eq:1}
\end{equation}
With $N_\kappa^i$ the number of SGD steps performs to reach the $\varepsilon-$accuracy and using Lemma \ref{lem:inner}, we have
\begin{equation}
    \left \| \boldsymbol\Omega_\kappa^\varepsilon- \boldsymbol\Omega_\kappa^i \right \|^2 \approx \left ( \mu N_\kappa^i \right)^2\left \| \nabla\mathcal{L}_\kappa^i \right \|^2
    \label{eq:2}
\end{equation}
Putting \eqref{eq:2} in \eqref{eq:1} and rearranging the terms, we obtain
\begin{equation}
    \left \| \nabla\mathcal{L}_\kappa^\varepsilon \right \|^2\leq \left ( \delta_1\mu N_\kappa^i-1 \right)^2\left \| \nabla\mathcal{L}_\kappa^i \right \|^2 +2\left \| \nabla\mathcal{L}_\kappa^\varepsilon \right \|\cdot \left \| \nabla\mathcal{L}_\kappa^i \right \|
    \label{eq:3}
\end{equation}

On the other hand, the $\varepsilon-$accuracy implies that \mbox{$\left \| \nabla\mathcal{L}_\kappa^\varepsilon \right \|^2\ < \varepsilon$}. Then, using Lemma \ref{lem:xbound}.1 and setting \mbox{$ G^i = \left \|\nabla_{\boldsymbol\Omega_k^i}\mathcal L_k^i \right \|\cdot \left \|\nabla_{\boldsymbol\Omega_k^\varepsilon}\mathcal L_k^\varepsilon \right \|$} in \eqref{eq:3}, we have 
\begin{equation}
    \left ( \delta_1\mu N_\kappa^i-1 \right)^2\left \| \nabla\mathcal{L}_\kappa^i \right \|^2 +2\left \| \nabla\mathcal{L}_\kappa^\varepsilon \right \|\cdot \left \| \nabla\mathcal{L}_\kappa^i \right \| < \varepsilon
    \label{eq:4}
\end{equation}
We obtain the desired result by rearranging the terms in \eqref{eq:4}.
\subsection{Proof of Corollary \ref{cor:steps}}
\label{apdx:cor2}
This proof of corollary is straightforward. Indeed, with $i=0$ for the random initialization scheme (RI) and $i=m$ for the meta-initialization (MI), the  $\varepsilon-$accuracy is achieved after $ N_\kappa^0$ and $ N_\kappa^m$, respectively. We obtain the desired result by differencing the two expressions and rearranging the terms. 

\section{}
\label{apdx:prop3}
1) In Corollary \ref{cor:steps}, it is obvious that for $N_\kappa^m$ to be less than $N_\kappa^0$, say $N_\kappa^m < N_\kappa^0$, the right side of the inequality has to be less than zero. That is:
\begin{equation*}
    \varepsilon \Delta_2 -2G^o\left \|\nabla\mathcal L_k^m \right \|\Delta_1 < 0
\end{equation*}
By expanding and simplifying the terms, and given Assumption \ref{ass:vbound}.b we obtain the following:
\begin{equation*}
    \left \|\nabla\mathcal L_k^m \right \| +\left \|\nabla\mathcal L_k^0 \right \| \leq \frac{2G^0}{\varepsilon}\left \|\nabla\mathcal L_k^m \right \| <  \frac{2\zeta\sqrt{\varepsilon }}{\varepsilon}\left \|\nabla\mathcal L_k^m \right \| 
\end{equation*}
Now considering Lemma \ref{lem:xbound}.2 we can rewrite the previous inequality as:
\begin{equation*}
    \left \|\nabla\mathcal L_k^m \right \| < \sqrt{\varepsilon_m }\leq  \frac{\varepsilon }{2\zeta \sqrt{\varepsilon }-\varepsilon }\zeta 
\end{equation*}
Squaring this expression yields the result in Proposition \ref{prop:prop3}.1\\

2) The term $\Delta_1$ in  Corollary \ref{cor:steps}, with Assumption  \ref{ass:smoothness}.a yields the following:
\begin{equation}
    \left | \Delta_1  \right |^2\leq \delta_1^2\left \| \boldsymbol\Omega_\kappa^m- \boldsymbol\Omega_\kappa^0 \right \|^2
    \label{eq:5}
\end{equation}
Using Lemma \ref{lem:outer} and Assumption \ref{ass:vbound}, \eqref{eq:5} can be rewritten as:
\begin{equation}
    \left | \Delta_1  \right |^2 < \left ( \delta_1\eta R_\kappa\zeta  \right )^2
    \label{eq:6}
\end{equation}

Similarly, with Assumption  \ref{ass:smoothness}.b and Lemma \ref{lem:outer}, the term $\Delta_2$ satisfies
\begin{equation}
    \left | \Delta_2  \right |^2 < \left ( \delta_2\eta R_\kappa\zeta  \right )^2
    \label{eq:7}
\end{equation}

Combining \eqref{eq:6} and \eqref{eq:7} with Corollary \ref{cor:steps}, the right side of the inequality in Corollary \ref{cor:steps}, denoted Q satisfies
\begin{equation*}
    Q \leq\frac{\eta R_\kappa\zeta}{\left ( \delta_1\Upsilon \mu \right )^2}\left ( \delta _2\varepsilon +2G^0\delta_1 \left \|\nabla\mathcal L_k^m \right \|\right )
\end{equation*}
With $\left \|\nabla\mathcal L_k^m \right \|^2<\varepsilon_m$ we get and \mbox{$G^0<\zeta\sqrt{\varepsilon}$}
\begin{equation}
    (N_\kappa^m)^2- (N_\kappa^0)^2 < Q <  \frac{\eta R_\kappa\zeta }{\left ( \delta_1\Upsilon \mu \right )^2}\left ( \delta _2\varepsilon +2\zeta\delta_1 \sqrt{\varepsilon _m\varepsilon}\right ).
\end{equation}

3) Assumption \ref{ass:lipsz} and Lemma \ref{lem:inner} hold the following
\begin{equation}
    \left | \mathcal{L}_\kappa^\varepsilon -  \mathcal{L}_\kappa^0\right |\leq \gamma \mu N_\kappa^0\left \| \nabla\mathcal L_k^0 \right \| < \gamma \mu N_\kappa^0\zeta
    \label{eq:7}
\end{equation}
Similarly, we have
\begin{equation}
    \left | \mathcal{L}_\kappa^\varepsilon -  \mathcal{L}_\kappa^m\right |\leq \gamma \mu N_\kappa^m\left \| \nabla\mathcal L_k^m \right \| < \gamma \mu N_\kappa^m\sqrt{\varepsilon _m}
    \label{eq:8}
\end{equation}
Then, combining  \eqref{eq:7} and \eqref{eq:8} yields the result in Proposition \ref{prop:prop3}.3\\







 


\bibliographystyle{IEEEtran}
\bibliography{references}


 

\end{document}